\documentclass[12pt]{article}
\pdfoutput=1
\usepackage{jheppub}
\usepackage{mathtools,amsmath,amssymb,amsthm,physics,bm,float}
\usepackage[utf8]{inputenc}
\usepackage{graphicx}
\usepackage{enumitem}
\usepackage{hyperref}

\begin{document}

\DeclarePairedDelimiter{\ceil}{\lceil}{\rceil}
\newtheorem{theorem}{Theorem}
\newtheorem{lemma}[theorem]{Lemma}
\newtheorem{corollary}[theorem]{Corollary}
\newcommand{\be}{\begin{equation}}
\newcommand{\ee}{\end{equation}}
\newcommand{\bi}{\begin{itemize}}
\newcommand{\ei}{\end{itemize}}
\newcommand{\mathbbm}[1]{\text{\usefont{U}{bbm}{m}{n}#1}}
\def\ba#1\ea{\begin{align}#1\end{align}}
\def\bg#1\eg{\begin{gather}#1\end{gather}}
\def\bm#1\em{\begin{multline}#1\end{multline}}
\def\bmd#1\emd{\begin{multlined}#1\end{multlined}}
\setlength{\intextsep}{-1ex} 

\def\XD#1{{\color{magenta}{ [#1]}}}
\def\SAM#1{{\color{red}{ [#1]}}}
\def\WW#1{{\color{blue}{ [#1]}}}

\def\a{\alpha}
\def\b{\beta}
\def\c{\chi}
\def\C{\Chi}
\def\d{\delta}
\def\D{\Delta}
\def\e{\epsilon}
\def\ve{\varepsilon}
\def\g{\gamma}
\def\G{\Gamma}
\def\h{\eta}
\def\k{\kappa}
\def\l{\lambda}
\def\L{\Lambda}
\def\m{\mu}
\def\n{\nu}
\def\p{\phi}
\def\P{\Phi}
\def\vp{\varphi}
\def\q{\theta}
\def\Q{\Theta}
\def\r{\rho}
\def\s{\sigma}
\def\S{\Sigma}
\def\t{\tau}
\def\u{\upsilon}
\def\U{\Upsilon}
\def\w{\omega}
\def\W{\Omega}
\def\x{\xi}
\def\X{\Xi}
\def\y{\psi}
\def\Y{\Psi}
\def\z{\zeta}

\newcommand{\la}{\label}
\newcommand{\ci}{\cite}
\newcommand{\re}{\ref}
\newcommand{\er}{\eqref}
\newcommand{\se}{\section}
\newcommand{\sse}{\subsection}
\newcommand{\ssse}{\subsubsection}
\newcommand{\fr}{\frac}
\newcommand{\na}{\nabla}
\newcommand{\pa}{\partial}
\newcommand{\td}{\tilde}
\newcommand{\wtd}{\widetilde}
\newcommand{\ph}{\phantom}
\newcommand{\eq}{\equiv}
\newcommand{\wg}{\wedge}
\newcommand{\cd}{\cdots}
\newcommand{\nn}{\nonumber}
\newcommand{\qu}{\quad}
\newcommand{\qqu}{\qquad}
\newcommand{\lt}{\left}
\newcommand{\rt}{\right}
\newcommand{\lra}{\leftrightarrow}
\newcommand{\ol}{\overline}
\newcommand{\ap}{\approx}
\renewcommand{\(}{\left(}
\renewcommand{\)}{\right)}
\renewcommand{\[}{\left[}
\renewcommand{\]}{\right]}
\newcommand{\<}{\langle}
\renewcommand{\>}{\rangle}
\newcommand{\Hc}{\mathcal{H}_{code}}
\newcommand{\HR}{\mathcal{H}_R}
\newcommand{\HRb}{\mathcal{H}_{\ol{R}}}
\newcommand{\lan}{\langle}
\newcommand{\ran}{\rangle}
\newcommand{\Hra}{\mathcal{H}_{W_\alpha}}
\newcommand{\Hrba}{\mathcal{H}_{\ol{W}_\alpha}}

\newcommand{\bH}{{\mathbb H}}
\newcommand{\bR}{{\mathbb R}}
\newcommand{\bZ}{{\mathbb Z}}
\newcommand{\cA}{{\mathcal A}}
\newcommand{\cB}{{\mathcal B}}
\newcommand{\cC}{{\mathcal C}}
\newcommand{\cE}{{\mathcal E}}
\newcommand{\cI}{{\mathcal I}}
\newcommand{\cN}{{\mathcal N}}
\newcommand{\cO}{{\mathcal O}}
\newcommand{\zb}{{\bar z}}

\newcommand{\Area}{\operatorname{Area}}
\newcommand{\ext}{\operatorname*{ext}}
\newcommand{\total}{\text{total}}
\newcommand{\bulk}{\text{bulk}}
\newcommand{\brane}{\text{brane}}
\newcommand{\matter}{\text{matter}}
\newcommand{\Wald}{\text{Wald}}
\newcommand{\anomaly}{\text{anomaly}}
\newcommand{\extrinsic}{\text{extrinsic}}
\newcommand{\gen}{\text{gen}}
\newcommand{\mc}{\text{mc}}

\newcommand{\T}[3]{{#1^{#2}_{\ph{#2}#3}}}
\newcommand{\Tu}[3]{{#1_{#2}^{\ph{#2}#3}}}
\newcommand{\Tud}[4]{{#1^{\ph{#2}#3}_{#2\ph{#3}#4}}}
\newcommand{\Tdu}[4]{{#1_{\ph{#2}#3}^{#2\ph{#3}#4}}}
\newcommand{\Tdud}[5]{{#1_{#2\ph{#3}#4}^{\ph{#2}#3\ph{#4}#5}}}
\newcommand{\Tudu}[5]{{#1^{#2\ph{#3}#4}_{\ph{#2}#3\ph{#4}#5}}}

\newcommand{\bs}{\boldsymbol}
\newcommand{\bfRho}{\bs{\rho}}
\newcommand{\bfS}{\textbf{S}}

\title{
Replica Wormholes and Holographic Entanglement Negativity
}
\author{Xi Dong,}
\author{Sean McBride,}
\author{and Wayne W. Weng}
\affiliation{Department of Physics, University of California, Santa Barbara, CA 93106, USA}
\emailAdd{xidong@ucsb.edu}
\emailAdd{seanmcbride@ucsb.edu}
\emailAdd{wweng@ucsb.edu}

\abstract{
Recent work has shown how to understand the Page curve of an evaporating black hole from replica wormholes.  However, more detailed information about the structure of its quantum state is needed to fully understand the dynamics of black hole evaporation.  Here we study entanglement negativity, an important measure of quantum entanglement in mixed states, in a couple of toy models of evaporating black holes.  We find four phases dominated by different types of geometries: the disconnected, cyclically connected, anti-cyclically connected, and pairwise connected geometries.  The last of these geometries are new replica wormholes that break the replica symmetry spontaneously.  We also analyze the transitions between these four phases by summing more generic replica geometries using a Schwinger-Dyson equation.  In particular, we find enhanced corrections to various negativity measures near the transition between the cyclic and pairwise phase.
}
\maketitle

\section{Introduction}

Replica wormholes have played an important role in recent progress on solving the black hole information problem~\cite{Penington:2019kki,Almheiri:2019qdq}.  These wormholes appear as nontrivial saddle points that could dominate gravitational path integrals with replicated boundary conditions.  Their appearance leads to nontrivial ``island'' contributions in the quantum extremal surface (QES) formula for gravitational entropy~\cite{Engelhardt:2014gca,Dong:2017xht,Penington:2019npb,Almheiri:2019psf}.

So far most of the discussion has been centered on the von Neumann entropy.  While obtaining the von Neumann entropy is a good first step, we need more detailed information about the quantum state --- such as more general measures of entanglement --- to fully solve the black hole information problem.

In this paper, we take a first step towards understanding the structure of entanglement in an evaporating black hole and its Hawking radiation by studying entanglement negativity and its R\'enyi generalizations in a couple of toy models.  Just as the von Neumann entropy is a measure of quantum entanglement in pure states, the negativity is an important measure of entanglement in generally mixed states.  Therefore, the negativity provides an interesting probe in diagnosing the structure of multipartite entanglement in systems such as an evaporating black hole.

To understand negativity intuitively, consider a general state on two subsystems that is described by a density matrix, and take its partial transpose on the second subsystem.  The partially transposed density matrix could have negative eigenvalues, and the degree to which the eigenvalues are negative is characterized by the negativity and logarithmic negativity.  Both of these negativity measures are entanglement monotones, and the logarithmic negativity provides an upper bound on the distillable entanglement~\cite{Vidal:2002zz,Plenio:2005cwa,Audenaert:2003}.  Negativity has been discussed in a number of interesting prior works~\cite{Calabrese:2012ew,Calabrese:2012nk,Rangamani:2014ywa,Calabrese:2014yza,Chaturvedi:2016rft,Chaturvedi:2016rcn,Jain:2017aqk,Jain:2017xsu,Malvimat:2017yaj,Kudler-Flam:2018qjo,Tamaoka:2018ned,Malvimat:2018txq,Kudler-Flam:2019wtv,Kusuki:2019zsp,Basak:2020bot,Basak:2020oaf,Kudler-Flam:2020xqu,Lu:2020jza,Basak:2020aaa,Shapourian:2020mkc,Dong:2021clv,KumarBasak:2021rrx,Vardhan:2021npf}.

We now give a short summary of this paper.

In Section~\re{sec:neg}, we review the definition and properties of the negativity and its R\'enyi generalizations. In Section~\re{sec:jt}, we start our study of negativity in a toy model of an evaporating black hole in Jackiw-Teitelboim (JT) gravity with an end-of-the-world (EOW) brane.  This is a slight generalization of the model studied in~\cite{Penington:2019kki}, with the system describing the Hawking radiation divided into two subsystems so as to study negativity.

As we tune the parameters of the model, we find a rich phase diagram for the negativities consisting of four phases (see Figure~\re{fig:phase1}).  Each of the four phases is dominated by a saddle-point geometry of JT gravity (or a set of saddle points).  For a black hole before the Page time, we find a phase dominated by a totally disconnected geometry, whereas after the Page time, we find three distinct phases depending on how we divide the radiation system into two subsystems: the first phase is dominated by a cyclically connected geometry (which is the replica wormhole of~\cite{Penington:2019kki}), the second dominated by an ``anti-cyclically'' connected geometry, and the third dominated by pairwise connected geometries that are in one-to-one correspondence with non-crossing pairings.  These pairwise connected geometries are new replica wormholes that spontaneously break the replica symmetry.  Their appearance agrees with the general discussions on holographic negativity in~\cite{Dong:2021clv}.

In Sections~\re{sec:resolvent} and~\re{sec:trans}, we study the behavior of negativities near the transitions between the four phases.  Near these phase transitions, more geometries than the four types described earlier could dominate the gravitational path integral for R\'enyi negativities, and we need to sum over them.  In order to obtain the negativity and logarithmic negativity (as well as related negativity measures such as the partially transposed entropy~\cite{Dong:2021clv,Tamaoka:2018ned}), we need to analytically continue in the replica number.  We achieve this by using the resolvent for the partially transposed density matrix to find its eigenvalue distribution (which we call the ``negativity spectrum'').  To calculate this ``negativity resolvent,'' we organize the sum over geometries into a Schwinger-Dyson equation, which is similar to the method used in~\cite{Penington:2019kki}.  We develop this method for negativity in Section~\re{sec:resolvent} and apply it to both a microcanonical ensemble and canonical ensemble in Section~\re{sec:trans}. 

When the black hole is in a microcanonical ensemble, the Schwinger-Dyson equation simplifies into a cubic equation for the negativity resolvent, leading to concrete results for the negativities near all phase transitions.  This is similar to the case of a random mixed state studied in~\cite{Shapourian:2020mkc}.

When the black hole is in a canonical ensemble, the gravitational calculation is technically more difficult.  As a result, we study each of the phase transitions separately, for we only need to sum over a subset of geometries near each transition.  Near the transition between the ``disconnected'' phase and ``pairwise'' phase, the Schwinger-Dyson equation again simplifies, this time into a quadratic equation for the negativity resolvent.  Near the transition between the ``cyclic'' phase and ``pairwise'' phase, it is difficult to solve the Schwinger-Dyson equation exactly, but we solve it approximately in the semiclassical, or $\b\to0$, limit.  From its solution, we find that the negativity spectrum near the phase transition consists of two branches, each of which is approximately a shifted thermal spectrum with a cutoff.  One branch consists of positive eigenvalues, and the other has negative eigenvalues.  From this we find enhanced corrections to various negativity measures near the phase transition.  In particular, a quantity known as the refined R\'enyi-2 negativity receives an $\cO(1/\sqrt{\b})$ correction, similar to the enhanced corrections to the von Neumann entropy at the Page transition~\cite{Penington:2019kki,Dong:2020iod,Marolf:2020vsi}, whereas other negativity measures such as the logarithmic negativity and the partially transposed entropy exhibit $\cO(1/\b)$ corrections, similar to what happens to R\'enyi entropies $S_n$ with $n<1$.

Moving beyond the JT gravity model, we study in Section~\re{sec:top} the behavior of negativities in a topological model of 2-dimensional gravity with EOW branes.  This is a slight generalization of the model of~\cite{Marolf:2020xie}, where we again divide the radiation system into two subsystems to study negativity.  We find the situation to be very similar to the case of a microcanonical ensemble in JT gravity described earlier.  In particular, the Schwinger-Dyson equation again simplifies into a cubic equation for the negativity resolvent, leading to concrete results for the negativities.

We end with some concluding remarks in Section~\re{sec:conclusion} and several appendices.  In Appendix~\re{app:dom}, we derive the set of dominant geometries in each of the phases and near phase transitions.  In Appendix~\re{app:trans}, we provide a more detailed analysis near the transition between the cyclic phase and pairwise phase in the canonical ensemble.  In Appendix~\re{app:renyi}, we study the R\'enyi entropies near the Page transition in a similar fashion and show that they exhibit corrections analogous to the corrections to the negativities near the phase transition.\\

Related works appeared recently and have some partial overlap with our results on the study of negativity in JT gravity in the microcanonical ensemble \cite{Kudler-Flam:2021efr} and the canonical ensemble \cite{Vardhan:2021npf,Vardhan:2021mdy}.

\section{Entanglement negativity and R\'enyi negativities}
\label{sec:neg}

The motivation for entanglement negativity comes from the Peres-Horodecki criterion \cite{Peres:1996dw, Horodecki:1996nc}, also known as the PPT (positive partial transpose) criterion for mixed states, which we review here. Consider a mixed state $\rho_{AB}$ defined on the product Hilbert space $\mathcal{H} = \mathcal{H}_A \otimes \mathcal{H}_B$. A state is separable if it can be written as
\be
\rho_{AB} = \sum_{i=1}^k p_i \rho_A^i \otimes \rho_B^i, \; \; \; \; \sum_{i=1}^k p_i = 1
\ee
for states $\rho_A^i$ and $\rho_B^i$ on $\mathcal{H}_A$ and $\mathcal{H}_B$, respectively. Separable states are classical mixtures of product states and thus do not contain quantum entanglement; inseparable states are said to be entangled.

We denote the algebra of operators on $\mathcal{H}_i$ by $\mathcal{A}_i$, and the space of linear maps from $\mathcal{A}_A$ to $\mathcal{A}_B$ by $\mathcal{L}(\mathcal{A}_A, \mathcal{A}_B)$. A map $\Lambda \in \mathcal{L}(\mathcal{A}_A, \mathcal{A}_B)$ is said to be positive if 
\be
\Lambda: \mathcal{A}_A \rightarrow \mathcal{A}_B
\ee
maps positive operators to positive operators, and is completely positive if for all nonnegative integer $n$,
\be
\Lambda_n \equiv \Lambda \otimes \mathbbm{I}: \mathcal{A}_A \otimes \mathcal{M}_n \rightarrow \mathcal{A}_B \otimes \mathcal{M}_n
\ee
is positive, where $\mathcal{M}_n$ denotes the algebra of $n \times n$ complex matrices.  For separable states, this condition is clearly satisfied when $\Lambda$ is a positive map, as $\left(\Lambda \otimes \mathbbm{I}\right) \left( \rho_A \otimes \rho_B \right) = \left( \Lambda \rho_A \right) \otimes \rho_B \geq 0$. For inseparable states, this no longer holds in general, so a good diagnostic of entanglement would be a positive but not completely positive map, such that entangled states would have negative eigenvalues under the action of $\Lambda \otimes \mathbbm{I}$.

The partial transpose is such a positive but not completely positive map. Consider a bipartite system $AB$ with an orthonormal basis $\{ \ket{a} \}$ on $A$ and $\{ \ket{b} \}$ on $B$. Given a density matrix $\rho_{AB}$ on $AB$, we define the partially transposed density matrix as\footnote{For reasons that will become clear shortly, in later sections we will rename the subsystems $A$ and $B$ as $R_1$ and $R_2$, and the partial transposition $T_B$ is thus called $T_2$.}
\be
\expval{a, b | \rho_{AB}^{T_B} | a', b'} = \expval{a, b' | \rho_{AB}| a', b} .
\ee
Acting on a reduced density matrix on $B$, the partial tranpose becomes the usual transpose which preserves the eigenvalues of the original reduced density matrix and is therefore a positive map. Acting on the full density matrix is not guaranteed to preserve positivity. As an example, take an EPR pair of two qubits $A$ and $B$. The partial transpose of its density matrix has eigenvalues $\{\frac{1}{2}, \frac{1}{2}, \frac{1}{2}, -\frac{1}{2}\}$. We therefore see that the partial transpose can be a useful tool for differentiating between separable and inseparable states.\footnote{For $2 \times 2$ and $2 \times 3$ matrices, the PPT criterion is both necessary and sufficient for the state to be separable. For systems of general dimension it is not sufficient, as bound entangled states have positive semidefinite partial transpose and therefore require further entanglement criteria to be distinguished from separable states.}

The entanglement negativity $\mathcal{N} (\rho)$ is defined as the sum of the absolute values of the negative eigenvalues of this partially transposed density matrix and can be variably written as
\be
\mathcal{N}(\rho_{AB}) = \frac{||\rho_{AB}^{T_B}||_1-1}{2} = \sum_i \frac{\abs{\lambda_i} - \lambda_i}{2} = \sum_{i:\lambda_i < 0} \abs{\lambda_i}.
\ee
Here $||X||_1 \equiv \Tr \abs{X} = \Tr \sqrt{X^\dagger X}$ is the Schatten 1-norm of a matrix $X$. We see why negativity is such an appealing entanglement measure, as it is computed directly from a trace, as opposed to a variational principle in the case of other entanglement measures. Note that as we are taking a trace, it does not matter which subsystem we take the partial trace over, so choosing $\rho_{AB}^{T_B}$ instead of  $\rho_{AB}^{T_A}$ is merely a convention. The logarithmic negativity is similarly defined by
\be
\mathcal{E}(\rho_{AB}) = \log \( \sum_i \abs{\lambda_i} \) = \log \big( 2 \mathcal{N}(\rho_{AB}) + 1 \big).
\label{eq:logneg}
\ee
The logarithmic negativity is an upper bound on the distillable entanglement, i.e., the asymptotic number of EPR pairs that can be extracted from a set of identically prepared $\rho_{AB}$ with local operations and classical communication (LOCC).

We can also write a R\'enyi version of negativity via
\be
\mathcal{N}_n(\rho_{AB}) =  \Tr \left[ \left(\rho_{AB}^{T_B}\right)^n\right].
\ee
There is a subtlety in the analytic continuation of the R\'enyi negativity. As the negativity is defined by the absolute value of the eigenvalues of the partially transposed density matrix and the R\'enyi negativity is defined without an absolute value, we need to define different analytic continuations for even and odd $n$ such that
\ba
\mathcal{N}_{2m}^{\textrm{(even)}} &=  \sum_i \abs{\lambda_i}^{2m}, \nonumber \\
\mathcal{N}_{2m-1}^{\textrm{(odd)}} &=  \sum_i \textrm{sgn}(\lambda_i)\abs{\lambda_i}^{2m-1}.
\ea
The logarithmic negativity is then obtained from the even analytic condition:
\be
\mathcal{E}(\rho_{AB}) = \lim_{m \rightarrow 1/2} \log \mathcal{N}_{2m}^{\textrm{(even)}}\left( \rho_{AB} \right).
\ee
In this paper, we will also be interested in a generalization of the R\'enyi negativities termed refined R\'enyi negativities, which are given by
\ba
S^{T_B(n,\textrm{even})}(\rho_{AB}) &= -n^2 \partial_n \left( \frac{1}{n} \log \mathcal{N}_{n}^{\textrm{(even)}}  \right), \nonumber \\
S^{T_B(n,\textrm{odd})}(\rho_{AB}) &= -n^2 \partial_n \left( \frac{1}{n} \log \mathcal{N}_{n}^{\textrm{(odd)}}  \right).
\ea

The refined R\'enyi negativities are inspired by the refined R\'enyi entropies defined in~\cite{Dong:2016fnf}. In particular, we will be interested in two measures that descend from the refined R\'enyi negativities. The first is the partially transposed entropy $S^{T_B}(\rho_{AB})$ of~\cite{Dong:2021clv,Tamaoka:2018ned}, defined as the $m\to1$ limit of the refined odd R\'enyi negativity:
\be
S^{T_B} (\rho_{AB}) = - \frac{1}{2}\lim_{m \rightarrow 1} \partial_m \log \mathcal{N}_{2m-1}^{\textrm{(odd)}} = - \sum_i \lambda_i \log \abs{\lambda_i}.
\ee
$S^{T_B}$ is so named in analogy with the von Neumann entropy.  The other measure is the refined R\'enyi-2 negativity $S^{T_B(2,\textrm{even})}(\rho_{AB})$, which can be written as
\be
S^{T_B(2,\textrm{even})}(\rho_{AB}) = -\lim_{m \rightarrow 1} m^2 \partial_m \left( \frac{1}{m} \log \mathcal{N}_{2m}^{\textrm{(even)}} \right) = - \sum_i \frac{\lambda_i^2}{\sum_j \lambda_j^2} \log\( \frac{\lambda_i^2}{\sum_j \lambda_j^2}\).
\label{eq:secondrefined}
\ee
We will refer to this quantity as $S^{T_B(2)}$ for short. It is equivalent to the von Neumann entropy of $\(\rho_{AB}^{T_B}\)^2/\Tr\[\(\rho_{AB}^{T_B}\)^2\]$.

\section{The model and four phases}\label{sec:jt}

\subsection{JT gravity with EOW branes}
We start by reviewing the simple model of black hole evaporation studied in~\cite{Penington:2019kki} (see also~\cite{Kourkoulou:2017zaj,Almheiri:2018xdw}). This model consists of a black hole in JT gravity, decorated with an end-of-the-world (EOW) brane with tension $\mu$. The action is given by
\be
I = I_{\text{JT}} + \mu \int_\text{brane} dy,
\label{eq:action}
\ee
with the JT action being
\be
I_{\text{JT}} = -\frac{S_0}{2\pi} \left[ \frac{1}{2} \int_\mathcal{M} R + \int_\mathcal{\partial M} K \right] - \left[ \frac{1}{2} \int_\mathcal{M} \phi (R+2) + \int_\mathcal{\partial M} \phi K \right].
\ee
We have set $G_N = 1$, though it can be restored by sending the inverse temperature $\beta \rightarrow \beta G_N$. The parameter $S_0$ can be thought of as the zero temperature entropy of an eternal two-dimensional black hole. The EOW brane is endowed with $k$ orthonormal states, or ``flavors," which are entangled with an auxiliary reference system $R$. The states on the brane can be thought of as describing the interior partners of the early Hawking radiation $R$, so by increasing $k$ we can probe later regimes of an ``evaporating'' black hole. 

The entangled state of the black hole system $B$ and the ``radiation'' $R$ can be written as
\be
\ket{\Psi} = \frac{1}{\sqrt{k}} \sum_{i=1}^k \ket{\psi_i}_B \ket{i}_R.
\ee
The density matrix of the $R$ subsystem is therefore
\be
\rho_R = \frac{1}{k} \sum_{i,j=1}^k \ket{j}\bra{i}_R \braket{\psi_i}{\psi_j}_B.
\ee
The inner product $\braket{\psi_i}{\psi_j}_B$ is given by a gravitational path integral with standard Dirichlet boundary conditions on an asymptotic boundary interval and Neumann boundary conditions on the EOW branes:
\be
ds^2|_{\partial \mathcal{M}} = \frac{1}{\epsilon^2} d\tau^2, \; \; \; \; \phi = \frac{1}{\epsilon}, \; \; \; \; \epsilon \rightarrow 0 \nn \\
\ee
\be
\partial_n \phi|_{\textrm{brane}} = \mu, \; \; \; \; K = 0
\label{eq:bounds}
\ee
as well as specifying the brane states $i$ and $j$ at the endpoints of the EOW brane. As was shown in~\cite{Penington:2019kki}, while naively $\braket{\psi_i}{\psi_j} \propto \delta_{ij}$, this should be understood as an expectation value in an ensemble, and wormhole contributions indicate exponentially small fluctuations of the inner product. We illustrate the boundary conditions for the matrix elements of $\rho_R$ as follows:
\be
\includegraphics[width=.5\textwidth]{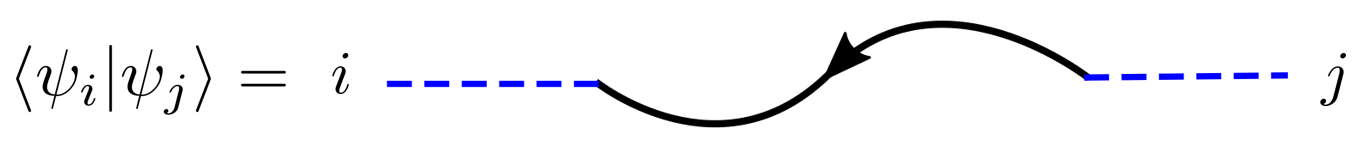}
\label{eq:braneBC}
\ee
The solid black line denotes an asymptotic boundary interval for the gravitational path integral, while the blue dashed lines are index lines that impose boundary conditions for the brane states. Computing $\Tr\( \rho_R \)$ means contracting the open index lines and summing over all possible geometries respecting the boundary conditions \eqref{eq:bounds}.

To study negativity, we need to consider a bipartite mixed state. To that end, we split the radiation system into two subsystems $\mathcal{H}_R = \mathcal{H}_{R_1} \otimes \mathcal{H}_{R_2}$ consisting of $k_1$ and $k_2$ states, respectively, such that $k = k_1 k_2$ and 
\be
\rho_{R_1 R_2} = \frac{1}{k}\sum_{i_1,i_2=1}^{k_1} \sum_{j_1,j_2=1}^{k_2} |i_1,j_1 \rangle \langle i_2,j_2| \langle\psi_{i_2,j_2}|\psi_{i_1,j_1}\rangle.
\ee
We will refer to this partitioned density matrix as $\rho_R$ from now on. We define our partially transposed density matrix as the partial transpose over $R_2$, i.e.,
\be
\rho_{R_1 R_2}^{T_{R_2}} =\frac{1}{k}\sum_{i_1,i_2=1}^{k_1} \sum_{j_1,j_2=1}^{k_2} |i_1,j_2 \rangle \langle i_2,j_1| \langle\psi_{i_2,j_2}|\psi_{i_1,j_1}\rangle.
\ee
We will use the shorthand $\rho_R^{T_2}$ moving forward. This partial transpose affects the boundary conditions for our path integral by swapping the brane flavor index lines corresponding to states in $\mathcal{H}_{R_2}$. The resulting boundary conditions are illustrated in Figure~\ref{fig:rhoBCs}.
\begin{figure}
    \centering
    \includegraphics[width=0.9\textwidth]{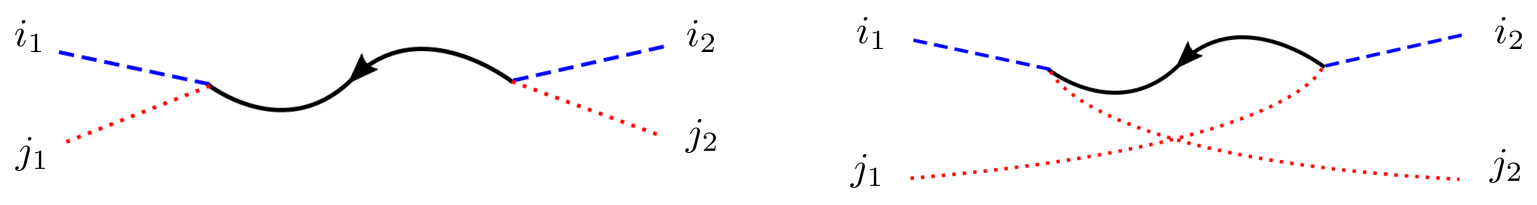}
    \caption{Boundary conditions for $\rho_{R_1R_2}$ and $\rho_{R_1R_2}^{T_2}$. Blue (dashed) lines denote states in $R_1$, and red (dotted) lines denotes states in $R_2$. If we take a trace, these two boundary conditions are equivalent.}
    \label{fig:rhoBCs}
\end{figure}

\subsection{Dominant saddles} \label{sec:domsaddles}
As in any calculation with a gravitational path integral, our first task is to identify the saddle-point geometries which obey the given boundary conditions and sum over them with the appropriate weight. As our goal is to compute R\'enyi negativities $\Tr \[\( \rho_R^{T_2} \)^n\]$, our boundary conditions will consist of $n$ copies of the boundary conditions illustrated on the right of Figure \ref{fig:rhoBCs}, with matching brane flavor indices contracted. The set of all classical saddles consists of oriented two-dimensional surfaces which end on the asymptotic boundaries and EOW branes, possibly connecting two or more boundaries.

\begin{figure}
\centering
\includegraphics[width=\textwidth]{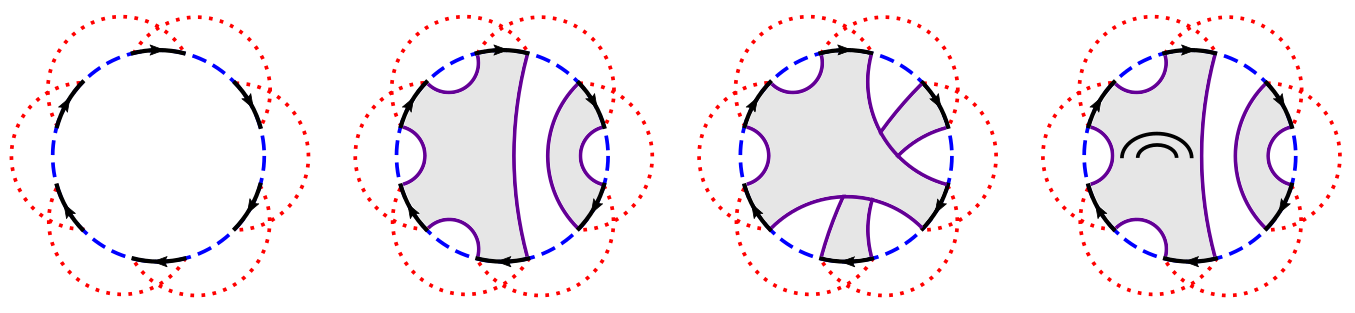}
\caption{Boundary conditions and some possible classical geometries that contribute to $\Tr [( \rho_R^{T_2} )^6]$. Index lines that run along EOW branes are not shown for visual clarity. Left: The partially transposed boundary conditions from Figure \ref{fig:rhoBCs} with brane flavor indices contracted. Center left and right: The geometries consist of disjoint unions of disks in non-crossing or crossing configurations. Right: The geometry has a single handle and will be suppressed by a factor of $e^{-2S_0}$ relative to the first geometry.} 
\label{fig:exDiagrams}
\end{figure}

As our gravitational action \eqref{eq:action} is independent of brane flavor, we can factorize the flavor contributions so that
\be
Z_{\textrm{saddle}} = Z_{\textrm{grav}} f(k_1,k_2)
\ee
for some function $f$ of the brane Hilbert space dimensions $k_1$ and $k_2$. The gravitational partition function $Z_n$ for a surface connecting $n$ boundaries depends on the Euler characteristic $\chi$ of the surface in the schematic form
\be
Z_n \sim e^{S_0 \chi}.
\label{eq:Zn1}
\ee
The contribution of a surface with genus $g \geq 1$ is therefore suppressed by $e^{-2gS_0}$ for large $S_0$. This means that the only classical geometries we need to consider are disks or disjoint unions of disks, and $Z_{\textrm{grav}}$ is a product of disk partition functions $Z_n$. We will therefore assume $S_0 \gg 1$ throughout the rest of the paper. We illustrate some examples of these disk geometries as well as a higher genus geometry in Figure~\ref{fig:exDiagrams}.

More precisely, for a disk connecting $n$ boundaries in JT gravity we have 
\begin{equation}
Z_n = e^{S_0}\int_0^\infty ds \rho(s) y(s)^n, \quad y(s) \equiv e^{-\frac{\beta s^2}{2}}2^{1-2\mu} \abs{\Gamma \left( \mu - \frac{1}{2} + is \right)}^2
\label{eq:Zn}
\end{equation}
where $\rho(s) = \frac{s}{2\pi^2} \sinh\(2\pi s\)$ is the disk density of states\footnote{The density of states is more typically written in the $E = s^2$ energy basis such that $\rho(E) = \frac{1}{4\pi^2} \sinh \left( 2 \pi \sqrt{E}\right)$. Here $s$ can be thought of as an entropy, and is the more natural variable for our purposes.} in JT gravity. In order to recover \eqref{eq:Zn1}, we take $S_0 \gg 1$ while keeping other parameters fixed. We emphasize this is a schematic approximation that should only be used to motivate the pertinent saddles for our problem; in general, there are parametric corrections from the full expression of $Z_n$ which will be discussed in more detail in Sections~\ref{sec:resolvent} and~\ref{sec:trans}.

Since we are ignoring higher genus surfaces, the sum over geometries with $n$ replicated asymptotic boundaries is equivalent to a sum over elements of $S_n$, the permutation group on $n$ elements. In particular, the sum takes the form 
\ba
\Tr\[\(\rho_R^{T_2}\)^n\] &= \frac{1}{(k Z_1)^n} \sum_{g \in S_n} \(\prod_{i=1}^{\chi(g)} Z_{|c_i(g)|}\) k_1^{\chi(g^{-1} X)} k_2^{\chi(g^{-1} X^{-1})} \label{eq:renyi_neg_exact} \\
&\sim \frac{1}{(ke^{S_0})^n} \sum_{g \in S_n} \left( e^{S_0}\right)^{\chi(g)} k_1^{\chi(g^{-1} X)} k_2^{\chi(g^{-1} X^{-1})}, \label{eq:renyi_neg_schem}
\ea
where $\chi(g)$ is the number of disjoint cycles of the permutation $g$, $|c_i(g)|$ is the length of the $i$-th disjoint cycle of $g$, and $X$ ($X^{-1}$) is the (anti-)cyclic permutation of length $n$. Unless otherwise specified, we will take $k, e^{S_0} \gg 1$. Note
\be
\chi(\mathbbm{1}) = n, \quad \chi(X) = \chi(X^{-1}) = 1.
\ee

\begin{figure}
    \centering
    \includegraphics[width=\textwidth]{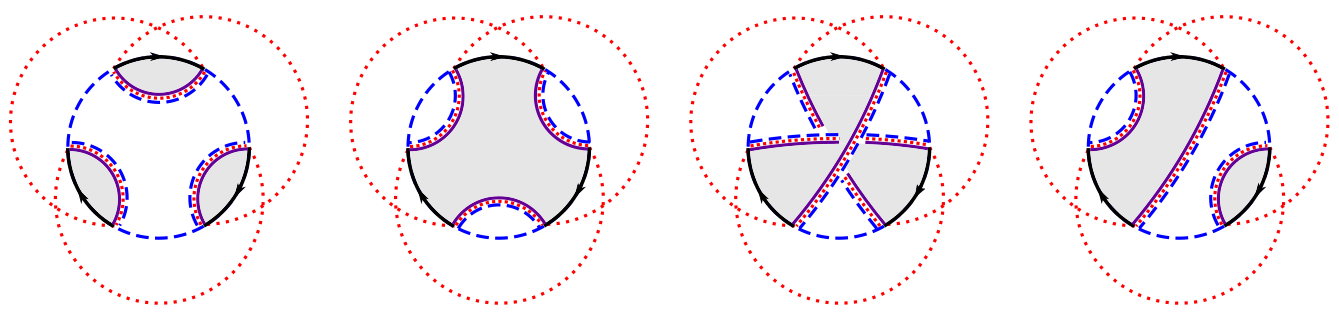}
    \caption{The four classes of geometries which dominate the R\'enyi negativity calculation. In order, they are the disconnected ($g = \mathbbm{1}$), cyclic ($g = X$), anti-cyclic ($g= X^{-1}$), and pairwise ($g=\tau$) geometries. The pairwise geometries spontaneously break replica symmetry. The black lines are oriented asymptotic boundaries, the purple lines are EOW branes, the blue (dashed) lines denote $k_1$ index loops, and the red (dotted) lines denote $k_2$ index loops.}
    \label{fig:geos}
\end{figure}

The sums~\eqref{eq:renyi_neg_schem} and ~\eqref{eq:renyi_neg_exact} over elements of the permutation group has a simple geometric interpretation. The permutation $g$ determines how the asymptotic boundaries are connected by EOW branes, while the powers of $k_1$ and $k_2$ count the number of index loops. The totally disconnected geometry corresponds to $g = \mathbbm{1}$, while the totally connected geometry corresponds to $g = X$. What does the $g = X^{-1}$ geometry look like? We show examples of these three geometries in Figure \ref{fig:geos}. Two of these geometries, the disconnected and cyclic geometries, belong in the class of non-crossing diagrams discussed by \cite{Penington:2019kki}. The third, the anti-cyclic geometry, is in some sense equivalent to the cyclic geometry if one reverses the orientation of the boundary, or equivalently if one exchanges $k_1$ and $k_2$. This statement will be explained in more detail in Section \ref{sec:antiplanar}.

One might naively guess that the anti-cyclic geometry would never dominate the R\'enyi negativity, for the same reasons as in~\cite{Penington:2019kki} where crossing partitions in the calculation of the R\'enyi entropy were suppressed by factors of $1/k^2$. In fact this geometry dominates in a very large parameter regime: as we show in Appendix~\ref{app:dom}, we have the following phases dominated by the corresponding permutation $g$:
\begin{alignat}{4}
&\textrm{Totally disconnected:} \qu && e^{S_0} \gg k_1 k_2 && \rightarrow && \qqu g = \mathbbm{1}  \nonumber \\
&\textrm{Cyclically connected:} \qu && k_1 \gg k_2  e^{S_0} && \rightarrow && \qqu g = X \nonumber  \\
&\textrm{Anti-cyclically connected:} \qu && k_2 \gg k_1 e^{S_0} && \rightarrow && \qqu g = X^{-1} \nn \\
&\textrm{Pairwise connected:} \qu && k_1 k_2 \gg e^{S_0}, && \rightarrow && \qqu g = \tau \\
& &&  e^{-S_0} \ll k_1/k_2 \ll e^{S_0} && && \nn
\end{alignat}

To see this intuitively, we calculate the contributions to $\Tr \[\( \rho_R^{T_2}\)^n\]$ from these geometries. Consider the totally disconnected phase dominated by $g = \mathbbm{1}$. We have $\chi(g)=n$ and $\chi(g^{-1}X) = \chi(g^{-1}X^{-1}) = 1$, so this diagram contributes schematically
\be
g = \mathbbm{1} \qu \Rightarrow \qu k \(e^{S_0}\)^n
\ee
to the sum in~\er{eq:renyi_neg_schem} as in~\cite{Penington:2019kki}. It is then unsurprising that this dominates in the parameter regime $e^{S_0} \gg k$, since it is the unique diagram which maximizes the power of $e^{S_0}$. Note that the contribution only depends on $k$, and not $k_1$ and $k_2$ individually.

Now, consider the cyclically connected phase with $g = X$. Then $\chi(g) = 1$, $\chi(g^{-1}X) = \chi(\mathbbm{1}) = n$, and
\be
\chi(g^{-1}X^{-1}) = \chi(X^{-2}) = f(n) \equiv \begin{cases}
    1, & \text{$n$ odd},\\
    2, & \text{$n$ even}.
\end{cases}
\ee
Hence, the cyclic diagrams contribute schematically
\begin{alignat}{3}
&g = X  && \qu \Rightarrow \qu && e^{S_0} k_1^n k_2^{f(n)}
\end{alignat}
to the sum in~\er{eq:renyi_neg_schem}. This configuration maximizes the power of $k_1$, and therefore it is expected to become important in the parameter regimes where $k_1$ is comparably large. In fact, as we prove in Appendix~\ref{app:dom}, it is the unique dominant diagram in the regime $k_1 \gg k_2 e^{S_0}$. Note that compared to the R\'enyi entropy calculation in~\cite{Penington:2019kki} (which can be recovered by setting $k_2=1$ here), the cyclic geometry is suppressed by $1/k_2^{n-f(n)}$. We will also show this diagrammatically in Section \ref{sec:resolvent}.

The anti-cyclically connected phase with $g = X^{-1}$ is similar, so we will not go through the analysis. In the end, the anti-cyclic diagrams contribute schematically
\begin{alignat}{3}
&g = X^{-1} && \qu \Rightarrow  \qu && e^{S_0} k_1^{f(n)} k_2^n
\end{alignat}
to the sum~\eqref{eq:renyi_neg_schem}. It is thus expected to become important in the parameter regimes where $k_2$ is comparably large, and we can prove that they are the unique dominant diagrams when $k_2 \gg k_2 e^{S_0}$.

Finally, there is one additional class of dominant geometries we should consider: the pairwise connected phase with $g = \tau$. As we show in Appendix~\ref{app:dom}, these diagrams dominate in a fourth regime satisfying both $k_1 k_2 \gg e^{S_0}$ and $e^{-S_0} \ll k_1/k_2 \ll e^{S_0}$, and they are the only diagrams aside from the disconnected, cyclically connected, and anti-cyclically connected diagrams that can dominate in a large regime of the parameter space. These geometries are in one-to-one correspondence with the set of permutations $\tau$ known as non-crossing pairings. For even $n$, a pairwise connected geometry is constructed by choosing an element in $\tau$, for example $(12)(34)\cdots(n-1,n)$, and connecting paired asymptotic boundaries by two-boundary wormholes. For odd $n$, the geometries are given by a similar non-crossing pairings of boundaries, plus a single one-boundary connected component. We show an example of such a geometry in Figure \ref{fig:geos}. It is evident that such geometries spontaneously break the replica symmetry.

As we show in Appendix~\ref{app:dom}, each pairwise connected geometry contributes schematically
\be
g = \tau \implies \left( e^{S_0} \right)^{\left\lceil \frac n2 \right\rceil}k^{\left\lfloor \frac n2 \right\rfloor + 1}
\ee
to the sum in~\er{eq:renyi_neg_schem}, where $\lceil\frac n2\rceil$ and $\lfloor \frac n2 \rfloor$ denote the ceiling and floor function, respectively. A pairwise connected diagram in some sense puts $k_1$, $k_2$, and $e^{S_0}$ on the most equal footing by maximizing the sum of the three exponents in \eqref{eq:renyi_neg_schem}. As with the disconnected geometry, the contribution of a pairwise connected geometry only depends on $k$, and not $k_1$ and $k_2$ individually.

\begin{figure}
    \centering
    \includegraphics[width=0.5\textwidth]{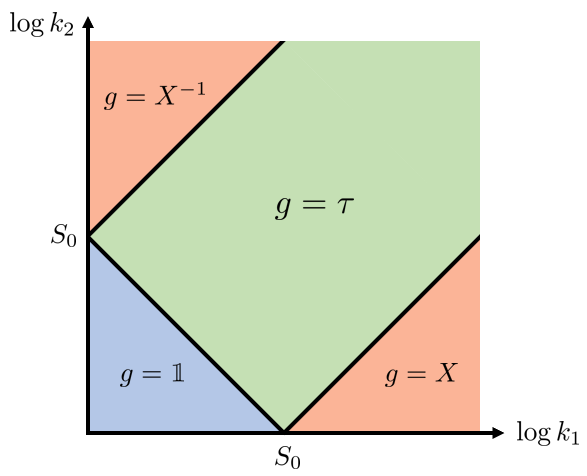}
    \caption{The phase diagram for entanglement negativity.  The four phases are labeled by the permutations corresponding to their dominant geometries.}
    \label{fig:phase1}
\end{figure}

\subsection{Contributions to negativities}

\begin{table}
\centering
\def\arraystretch{1.8}
\begin{tabular}{ |c||c|c|c|c|  }
 \hline
 \multicolumn{5}{|c|}{Negativities in Dominant Phases} \\
 \hline
 $g$& $\mathbbm{I}$ & $X$ & $X^{-1}$ &$\tau$\\
 \hline 
 $\mathcal{N}_{2m}^{\textrm{(even)}}$   & $\displaystyle \frac{1}{k^{2m-1}}$  &$ \displaystyle \frac{Z_{2m}}{k_2^{2m-2}Z_1^{2m}}$&   $\displaystyle \frac{Z_{2m}}{k_1^{2m-2}Z_1^{2m}}$&$\displaystyle \frac{C_m Z_2^m}{k^{m-1}Z_1^{2m}}$\\
 $\displaystyle \mathcal{N}_{2m-1}^{\textrm{(odd)}}$&   $\displaystyle \frac{1}{k^{2m-2}}$  & $\displaystyle \frac{Z_{2m-1}}{k_2^{2m-2}Z_1^{2m-1}}$   &$\displaystyle \frac{Z_{2m-1}}{k_1^{2m-2}Z_1^{2m-1}}$&$\displaystyle \frac{(2m-1) C_{m-1}Z_2^{m-1}}{k^{m-1}Z_1^{2m-2}}$\\
 $\mathcal{E}$ &0 & $\displaystyle \log k_2$&  $\displaystyle \log k_1$&$\displaystyle \frac{1}{2} \left( \log k - S_0  \right) + \log \frac{8}{3 \pi} $\\
 $S^{T_2}$    &$\displaystyle \log k$ & $\log k_2 +S_0$&  $\displaystyle \log k_1 + S_0$ & $\displaystyle \frac{1}{2} \left( \log k + S_0  \right) - \frac{1}{2}$ \\
 $S^{T_2(2)}$&  $\displaystyle \log k$  & $\displaystyle 2 \log k_2 + S_0$& $\displaystyle 2 \log k_1 + S_0$ &$\displaystyle \log k - \frac{1}{2}$\\
 \hline
\end{tabular}
\caption{Top two rows: R\'enyi negativities in the four dominant phases labeled by the corresponding permutation $g$.  Bottom three rows: schematic values of three special limits of the analytic continued R\'enyi negativities (where we have used $Z_n \sim e^{S_0}$). Here $C_m = \frac{1}{m+1} \binom{2m}{m}$ is the Catalan number which gives the number of non-crossing pairings.}
 \label{table}
\end{table}

Having worked out the dominant geometries in the four phases and how they contribute to the R\'enyi negativities schematically (i.e., using $Z_n \sim e^{S_0}$), we now write their contributions exactly using \er{eq:renyi_neg_exact}.  This calculation is straightforward to do, for both even and odd $n$.  We then analytically continue the resulting R\'enyi negativities and find the values of three special limits that we defined in Section~\ref{sec:neg}: the logarithmic negativity $\mathcal{E}$, partially transposed entropy $S^{T_2}$, and refined R\'enyi-2 negativity $S^{T_2(2)}$.  We collect these results for all four phases in Table \ref{table}.

From these results, we find a phase diagram for negativity, which we show in Figure~\ref{fig:phase1}.  Unlike the von Neumann entropy which only has a single phase transition at $k \sim e^{S_0}$, we see that there are two distinct types of phase transitions for negativity, one from the disconnected phase to the pairwise phase and one from the pairwise phase to the cyclic phase. The transition from the pairwise phase to the anti-cyclic phase is similar to the pairwise-to-cyclic transition under the exchange $k_1 \leftrightarrow k_2$.

\section{Resolvent equation for partial transpose}
\label{sec:resolvent}

Having analyzed the negativity measures deep within each of the four phases, we now turn our attention to the behavior of negativities near the phase transitions. Generally speaking, more geometries than the four types studied in the previous section could dominate the R\'enyi negativities near a phase transition, and we need to sum over them.  We would then need to analytically continue the resulting sum to find special limits such as the logarithmic negativity. This is technically difficult to do directly.

Instead, we study the resolvent for the partially transposed density matrix $\r_R^{T_2}$, which we refer to as the negativity resolvent or simply the resolvent.  From this resolvent, we then extract the negativity spectrum, i.e., the eigenvalue distribution of $\r_R^{T_2}$.  This allows us to calculate the R\'enyi negativities and their special limits.

In this section, we derive a self-consistent equation for calculating the resolvent. As we will show, the sum over dominant diagrams in a ``planar" regime reduces to a Schwinger-Dyson equation, which can be resummed. This allows us to write down a closed-form equation for the resolvent.

The negativity resolvent is defined in terms of the partially transposed density matrix as\footnote{The resolvent $R\(\lambda\)$ should not be confused with the subsystem $R$ describing the Hawking radiation.}
\be\label{eq:neg_resolvent_traced}
R\(\lambda\) = \Tr \( \fr{1}{\lambda\mathbb{I} - \rho_R^{T_2}} \).
\ee
From the resolvent, the eigenvalue spectrum for $\rho^{T_2}_R$, which we will denote by $D(\lambda)$, can be obtained by taking the discontinuity across the real axis as follows
\be
D(\lambda) = \lim_{\epsilon \to 0^+} \fr{1}{2\pi i}\big( R(\lambda-i\epsilon)- R(\lambda + i \epsilon)\big).
\ee
From this, we can compute the R\'enyi negativities via
\ba
\mathcal{N}_{2m}^{(\text{even})} &= \int d\lambda\,D\(\lambda\) |\lambda|^{2m} ,\\
\mathcal{N}_{2m-1}^{(\text{odd})} &= \int d\lambda\,D\(\lambda\) \text{sgn}\(\lambda\)|\lambda|^{2m-1},
\ea
from which all other negativity measures we consider can be derived.

It will be useful to consider \eqref{eq:neg_resolvent_traced} in the following matrix form:
\ba
R^{i_1i_2}_{j_1j_2}(\lambda) &= \(\fr{1}{\lambda\mathbb{I}-\rho^{T_2}_R}\)^{i_1i_2}_{j_1j_2} \\
&= \frac{\delta^{i_1i_2}\delta_{j_1j_2}}{\lambda} + \sum_{n=1}^\infty \frac{1}{\lambda^{n+1}}\Big(\(\rho^{T_2}_R\)^n\Big)^{i_1i_2}_{j_1j_2},
\ea
where we denote the $R_1$ subsystem by upper $i$-type indices and $R_2$ by lower $j$-type indices. In the last line, we have expanded the expression in a formal power series in $1/\lambda$. Each term in the series is given by the $n$-replicated density matrix $\(\rho^{T_2}_R\)^n$, which defines a boundary condition with $n$ asymptotic boundaries with the brane indices contracted:
\be\nonumber
\includegraphics[width=\textwidth]{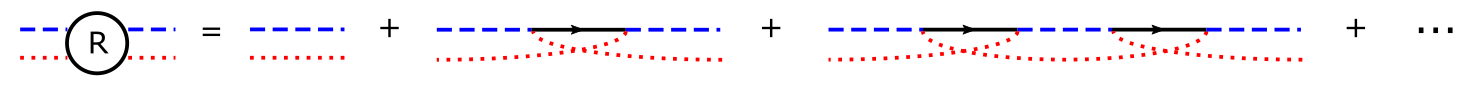}
\ee
where the blue (upper) dashed lines denote $i$-type index lines and red (lower) dotted lines denote $j$-type index lines. Each pair of blue/red index lines gives a factor of $1/\lambda$ and each asymptotic boundary gives a factor of $1/(kZ_1)$ coming from the normalization of the density matrix. Note that for up to two boundaries, the boundary conditions are the same after taking a final trace, with or without partial transpose.

In general, the gravitational path integral can be performed as a sum over bulk geometries satisfying the boundary conditions. In the JT model we introduced in Section~\ref{sec:jt}, higher genus corrections are highly suppressed by factors of $e^{-S_0} \ll 1$, so we only need to consider disjoint unions of disk geometries connecting any number of asymptotic boundaries. The path integral for the disk can be performed exactly including quantum corrections and is given by \eqref{eq:Zn}. For $n$ asymptotic boundaries, these disjoint unions of disk geometries are in one-to-one correspondence with elements of the permutation group $S_n$. As we show in Lemma~\re{lmplanar} and Corollary~\re{lmanti} of Appendix \ref{app:dom}, the only geometries that can possibly dominate in the limit $e^{S_0} \gg 1$ are the planar and anti-planar diagrams, which correspond to certain subsets of permutations in $S_n$. The term ``anti-planar" will be explained in detail in Section \ref{sec:antiplanar} and in Appendix~\ref{app:dom}. In fact, large regions of the phase diagram are dominated by either the planar or anti-planar geometries, which we will call the planar and anti-planar regimes. As we will now show, in each of these two regimes the resulting sum over diagrams can be resummed via a Schwinger-Dyson equation.

\subsection{Planar regime}

We now derive a resolvent equation in the parameter regime $k_2 \ll k_1 e^{S_0}$.  For reasons that will become clear shortly, we call it the planar regime.

Our strategy is to keep only the subset of geometries which have a possibility of dominating. As we outlined in Section~\ref{sec:domsaddles}, the phases in this regime away from phase transitions are defined by the disconnected, cyclic, and pairwise geometries. Closer to phase transitions, we expect more generic geometries which ``interpolate" between these three types of geometries to have a chance of dominating. Indeed, as we show in Lemma~\re{lmplanar} of Appendix \ref{app:dom}, the geometries which can dominate are precisely the planar diagrams, which correspond to the non-crossing partitions, i.e., the permutation group elements $g \in S_n$ lying on a geodesic between $\mathbbm{1}$ and $X$.

We can write the sum over planar geometries diagrammatically as
\be\nonumber
\includegraphics[width=\textwidth]{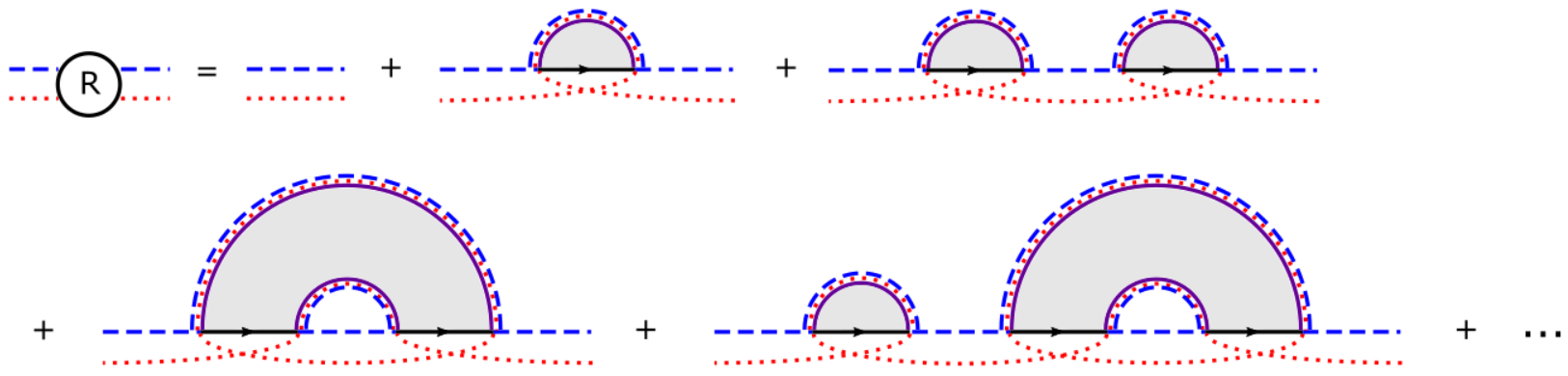}
\ee
This sum can be recast as a Schwinger-Dyson equation. Diagrammatically, we have
\be\nn
\includegraphics[width=\textwidth]{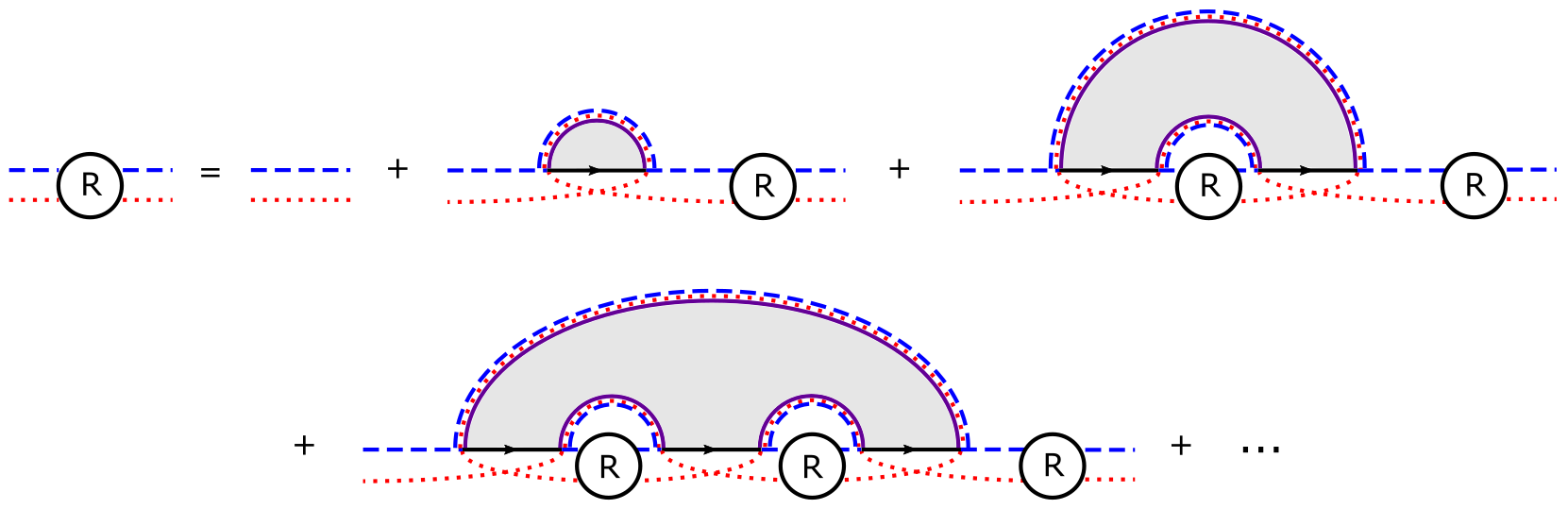}
\ee
or, as an equation,
\ba \label{eq:SD_planar}
R^{i_1i_2}_{j_1j_2} &= \frac{\delta^{i_1i_2}\delta_{j_1j_2}}{\lambda} + \frac{1}{\lambda} \sum_{n=1}^\infty \frac{Z_n}{\(kZ_1\)^n} \tilde{R}_{k_{n-1}k_1}\tilde{R}_{j_1k_2}\tilde{R}_{k_1k_3}\cdots \tilde{R}_{k_{n-3}k_{n-1}}R^{i_1i_2}_{k_{n-2}j_2},
\ea
where we have defined the partial trace of the resolvent matrix over the $R_1$ subsystem: $\tilde{R}_{j_1j_2} \equiv \sum_{i=1}^{k_1} R^{ii}_{j_1j_2}$, and repeated indices are summed over\footnote{Note that these repeated indices $k_1, k_2, \cd, k_{n-1}$ are $j$-type indices which should not be confused with the parameters $k_1$, $k_2$ that count the number of EOW brane states.}. Note that for $n =1$ the product of resolvents is simply $R_{j_1j_2}^{i_1i_2}$ and for $n = 2$ it is $\tilde{R}_{k_1k_1}R^{i_1i_2}_{j_1j_2}$. As is evident from the diagrammatics, the $i$-type indices denoting $R_1$ form simple self-contractions on all but the last resolvent, while the $j$-type indices denoting $R_2$ form a complicated set of contractions. Fortunately, this equation can be solved iteratively: starting with the leading solution $R^{i_1i_2}_{j_1j_2} = \delta^{i_1i_2} \delta_{j_1j_2}/\lambda + \mathcal{O}\(1/\lambda^2\)$, self-consistency then requires that $R^{i_1i_2}_{j_1j_2} \propto \delta^{i_1i_2} \delta_{j_1j_2}$ to all orders. As we will see, this allows us to rewrite the complicated product of resolvents as the following simple expression
\be \label{eq:prod_R}
\td{R}_{k_{n-1}k_1}\cdots \tilde{R}_{k_{n-3}k_{n-1}} R^{i_1i_2}_{k_{n-2}j_2} = \Bigg\{ \begin{array}{lll}
     & \(\fr{R}{k_2}\)^{n-1} R^{i_1i_2}_{j_1j_2}, &\quad n \text{ odd}, \\
     & k_2 \(\fr{R}{k_2}\)^{n-1} R^{i_1i_2}_{j_1j_2}, &\quad n \text{ even}.
\end{array}
\ee

Let us explain this in more detail. In general, the behavior for $n$ odd and even are different so we will need to treat these cases separately. To illustrate the simplification for the odd case, we first consider the contribution at $n=3$, which is the first non-trivial diagram under the partial transpose. We can write
\ba
\td{R}_{k_2k_1}\td{R}_{j_1k_2}R^{i_1i_2}_{k_1j_2} &= \(\td{R}_{j_1j_1}\)^2 R^{i_1i_2}_{j_1j_1} \delta_{k_2k_1} \delta_{j_1k_2} \delta_{k_1j_2} \nn \\
&= \(\td{R}_{j_1j_1}\)^2 R^{i_1i_2}_{j_1j_1} \delta_{j_1j_2} \nn \\
&= \(\td{R}_{j_1j_1}\)^2 R^{i_1i_2}_{j_1j_2}
\ea
where no summation on $j_1$ is implied. Now, recalling the full trace $R = \sum_{j} \tilde{R}_{jj}$, we find $\tilde{R}_{jj} = R/k_2$ and we can therefore write
\be
\td{R}_{k_2k_1}\td{R}_{j_1k_2}R_{k_1j_2}^{i_1i_2} = \(\fr{R}{k_2}\)^2 R_{j_1j_2}^{i_1i_2}.
\ee
We can understand the even case by looking at the first non-trivial contribution at $n = 4$. By a similar analysis as above, we have
\ba
\td{R}_{k_3k_1}\td{R}_{j_1k_2}\td{R}_{k_1k_3}R^{i_1i_2}_{k_2j_2} &= \(\td{R}_{j_1j_1}\)^3 R^{i_1i_2}_{j_1 j_1} \delta_{k_3k_1}\delta_{j_1k_2}\delta_{k_1k_3}\delta_{k_2j_2} \nn \\
&= k_2 \(\td{R}_{j_1j_1}\)^3 R^{i_1i_2}_{j_1j_2} \nn \\
&= k_2 \(\fr{R}{k_2}\)^3 R^{i_1i_2}_{j_1j_2}.
\ea
Note the additional factor of $k_2$ compared to $n=3$ due to the closed index loop formed from the first and third resolvent factors.

More generally, one can show that the even case always has a single index loop and the odd case has no index loops, leading to \eqref{eq:prod_R}. Using this, we can rewrite the Schwinger-Dyson equation \eqref{eq:SD_planar} as
\ba
&\lambda R^{i_1i_2}_{j_1j_2} = \delta^{i_1i_2}\delta_{j_1j_2} + k_2 \sum_{m = 1}^\infty \frac{Z_{2m-1}}{\(kk_2Z_1\)^{2m-1}}R^{2m-2} R^{i_1i_2}_{j_1j_2} + k_2^2 \sum_{m = 1}^\infty \frac{Z_{2m}}{\(kk_2Z_1\)^{2m}}R^{2m-1} R^{i_1i_2}_{j_1j_2}.
\nn \ea
Taking the full trace, we find
\be \label{eq:simple_SD_planar}
\lambda R = k + k_2\sum_{m = 1}^\infty \frac{Z_{2m-1}R^{2m-1}}{\(kk_2Z_1\)^{2m-1}} + k_2^2 \sum_{m = 1}^\infty \frac{Z_{2m}R^{2m}}{\(kk_2Z_1\)^{2m}}.
\ee
The gravitational partition function of the $n$-boundary totally connected geometry is given by \eqref{eq:Zn}, which we repeat here:
\begin{equation}
Z_n = e^{S_0}\int_0^\infty ds \rho(s) y(s)^n.
\end{equation}
Since this depends on $n$ only through $y(s)^n$, the sum over $n$ becomes a geometric series and \eqref{eq:simple_SD_planar} can be resummed into
\be \label{eq:resolvent_planar}
    \lambda R = k + k_2^2 e^{S_0} \int_0^{\infty} ds \rho(s) \frac{w(s)R(k+w(s)R)}{k^2 k_2^2 - w(s)^2 R^2},
\ee
where $w(s) \equiv y(s)/Z_1$. As a consistency check, when $k_2 = 1$ this reduces to the resolvent equation for the original (non-partially-transposed) density matrix derived in~\cite{Penington:2019kki}.

\subsection{Anti-planar regime} \label{sec:antiplanar}
We now consider a different parameter regime $k_1 \ll k_2 e^{S_0}$.  For reasons that will become clear shortly, we call it the anti-planar regime.  This anti-planar regime has a large overlap with the planar regime; the overlap region is $ e^{-S_0} \ll k_1 / k_2 \ll e^{S_0} $.  Together the two regimes cover the entire parameter space.

Once again, we will keep only the subset of geometries which have a possibility of dominating. As we outlined in Section~\ref{sec:domsaddles}, the phases in this regime away from phase transitions are defined by the disconnected, anti-cyclic, and pairwise geometries. Closer to phase transitions, we expect geometries which interpolate between these geometries to have a chance of dominating. Indeed, as we show in Corollary~\re{lmanti} of Appendix \ref{app:dom}, the geometries which can dominate are precisely the set of anti-planar diagrams, which are in one-to-one correspondence with permutations lying on a geodesic between $\mathbbm{1}$ and $X^{-1}$.  An example of an anti-planar geometry is the anti-cyclic geometry; two other (perhaps less obvious) examples are the disconnected and pairwise geometries.  Geometrically, anti-planar diagrams are precisely those that become planar diagrams after reversing the orientation of the asymptotic boundaries:
\be\nn
\includegraphics[width=0.7\textwidth]{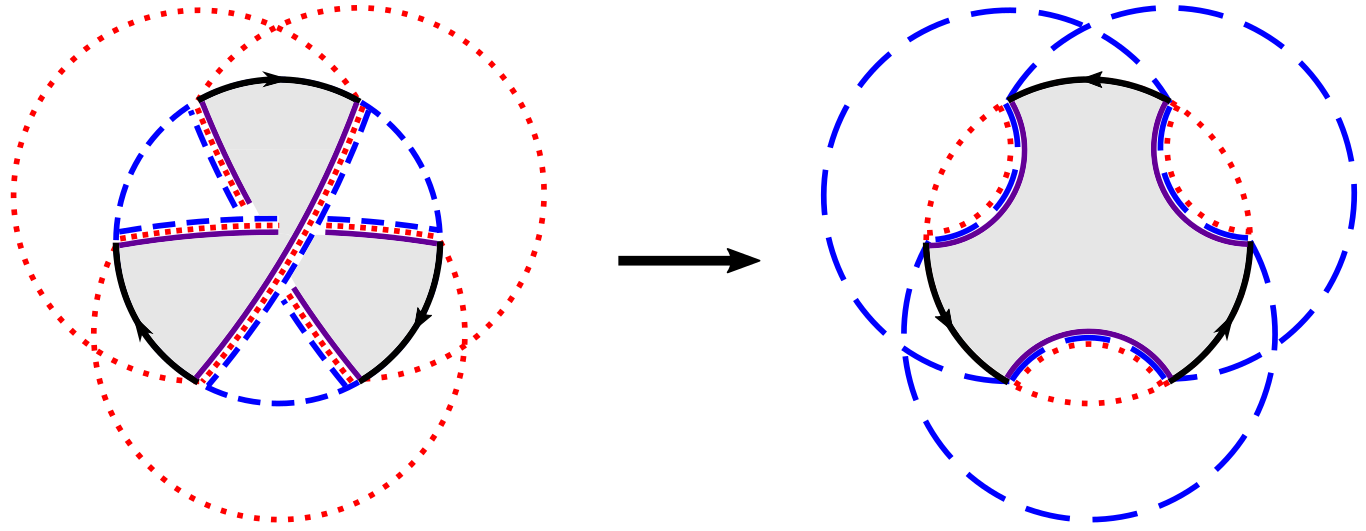}
\ee
The sum over anti-planar diagrams can similarly be recast as a Schwinger-Dyson equation. Diagrammatically, we have
\be\nn
\includegraphics[width=\textwidth]{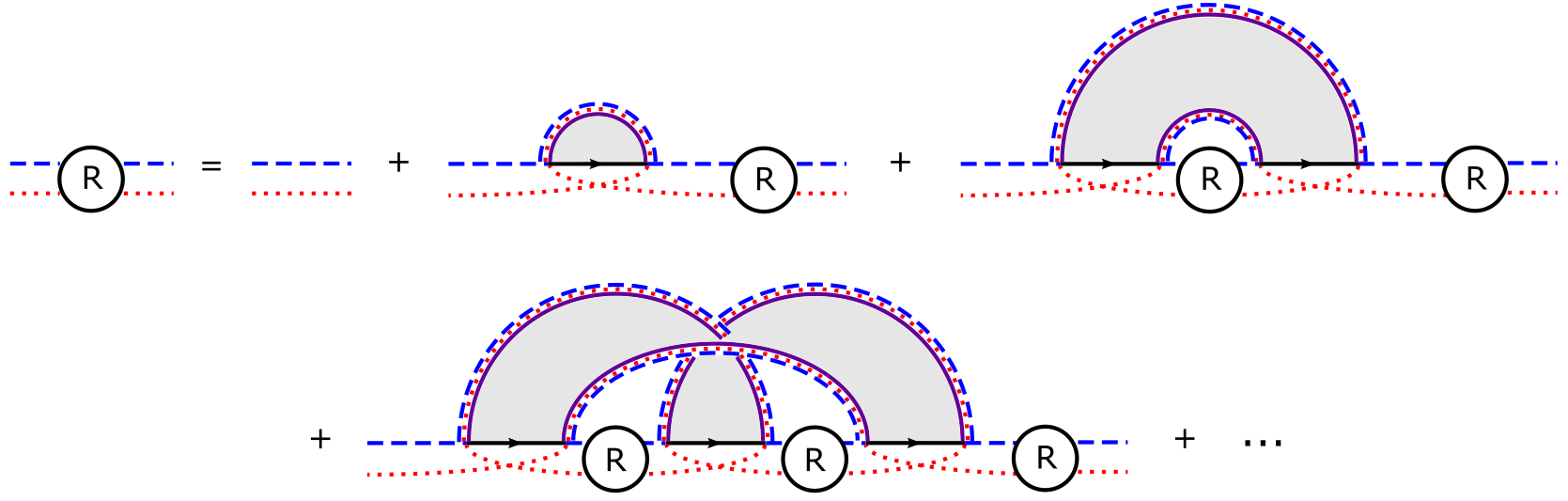}
\ee
or, as an equation,
\ba \label{eq:SD_antiplanar}
R^{i_1i_2}_{j_1j_2} &= \frac{\delta^{i_1i_2}\delta_{j_1j_2}}{\lambda} + \frac{1}{\lambda} \sum_{n=1}^\infty \frac{Z_n}{\(kZ_1\)^n} \tilde{R}^{k_{n-1}k_1}\tilde{R}^{i_1k_2}\tilde{R}^{k_1k_3}\cdots \tilde{R}^{k_{n-3}k_{n-1}}R_{j_1j_2}^{k_{n-2}i_2},
\ea
where we have defined the partial trace over the $R_2$ subsystem: $\tilde{R}^{i_1i_2} \equiv \sum_{j=1}^{k_2} R^{i_1i_2}_{jj}$. Note that the $n=1$ and $n=2$ terms are the same as in the planar case. Compared to the planar case, the $j$-type indices denoting $R_2$ now form simple self-contractions on all but the last resolvent, while the $i$-type indices denoting $R_1$ form the complicated set of contractions. In other words, the $i$-type indices now play the role of $j$-type indices in the planar regime, and vice-versa. As before, we can use the fact that $R^{i_1i_2}_{j_1j_2} \propto \delta^{i_1i_2} \delta_{j_1j_2}$ to all orders to rewrite the complicated product of resolvents as
\be
\td{R}^{k_{n-1}k_1}\cdots \tilde{R}^{k_{n-3}k_{n-1}} R_{j_1j_2}^{k_{n-2}i_2} = \Bigg\{ \begin{array}{lll}
     & \(\fr{R}{k_1}\)^{n-1} R^{i_1i_2}_{j_1j_2}, &\quad n \text{ odd}, \\
     & k_1 \(\fr{R}{k_1}\)^{n-1} R^{i_1i_2}_{j_1j_2}, &\quad n \text{ even}.
\end{array}
\ee
Using this, we can rewrite the Schwinger-Dyson equation \eqref{eq:SD_antiplanar} as
\ba
&\lambda R^{i_1i_2}_{j_1j_2} = \delta^{i_1i_2}\delta_{j_1j_2} + k_1 \sum_{m = 1}^\infty \frac{Z_{2m-1}}{\(kk_1Z_1\)^{2m-1}}R^{2m-2} R^{i_1i_2}_{j_1j_2} + k_1^2 \sum_{m = 1}^\infty \frac{Z_{2m}}{\(kk_1Z_1\)^{2m}}R^{2m-1} R^{i_1i_2}_{j_1j_2}.
\nn \ea
Taking the full trace, we find
\be \label{eq:simple_SD_antiplanar}
\lambda R = k + k_1\sum_{m = 1}^\infty \frac{Z_{2m-1}R^{2m-1}}{\(kk_1Z_1\)^{2m-1}} + k_1^2 \sum_{m = 1}^\infty \frac{Z_{2m}R^{2m}}{\(kk_1Z_1\)^{2m}}.
\ee
Finally, using \er{eq:Zn}, we resum \eqref{eq:simple_SD_antiplanar} into
\be
    \lambda R = k + k_1^2 e^{S_0} \int_0^{\infty} ds \rho(s) \frac{w(s)R(k+w(s)R)}{k^2 k_1^2 - w(s)^2 R^2}.
    \label{eq:Rint2}
\ee
As expected, this is simply the resolvent equation in the planar regime \eqref{eq:resolvent_planar} with the exchange $k_1 \leftrightarrow k_2$.

\section{Negativity spectrum near phase transitions}
\label{sec:trans}
Having derived the resolvent equation for the partial transpose, we now solve it to find the negativity spectrum near phase transitions.

For each negativity measure, we will be interested in a neighborhood near one of the phase transitions as we tune the relative sizes of our parameters, and we will analyze the corrections to the negativity measures listed in Table \ref{table}. Generically these corrections take the form of fluctuations about a fixed saddle point in the gravitational path integral. Near a phase transition, however, multiple saddles are competing for dominance, so enhanced corrections, i.e.\ corrections larger than those for any individual saddle, provide additional information about the entanglement structure of the system.

\subsection{Microcanonical ensemble}
\label{sec:MC}
Before studying the negativity spectrum in more detail, let us take a brief detour and consider the situation where the black hole is in a microcanonical ensemble in the JT model. In this case, we restrict to some small energy window $[s, s+\Delta s]$.  We write
\be
e^\bfS = \bfRho(s) \Delta s, \qquad \textbf{Z}_n = \bfRho(s) y(s)^n \Delta s, \qquad \textbf{w}(s) = \frac{y(s)}{\textbf{Z}_1} = e^{-\bfS},
\ee
where $\bfRho(s) \equiv e^{S_0} \rho(s)$ is the density of states. In this case, it can be shown that the sum over geometries for the R\'enyi negativities \eqref{eq:renyi_neg_exact} is given by
\ba \label{eq:microcanonical}
\Tr(\rho_R^{T_2})^n &= \frac{1}{(k e^\bfS)^n} \sum_{g \in S_n} \(e^\bfS\)^{\chi(g)} k_1^{\chi(g^{-1} X)} k_2^{\chi(g^{-1} X^{-1})}.
\ea
We recognize this as the $n^\textrm{th}$ R\'enyi negativity of a Wishart matrix with $e^\bfS$ degrees of freedom. This implies that $\rho^{T_2}_R$ in a microcanonical ensemble can be thought of as the partial transpose of a random matrix drawn from the Wishart distribution. A similar expression for the moments of the partial transpose of a random mixed state was derived in~\cite{Shapourian:2020mkc}. The similarity between the microcanonical JT model and a random mixed state was also noted in~\cite{Kudler-Flam:2021efr}.

\begin{figure}
\centerline{
\includegraphics[width=0.5\textwidth]{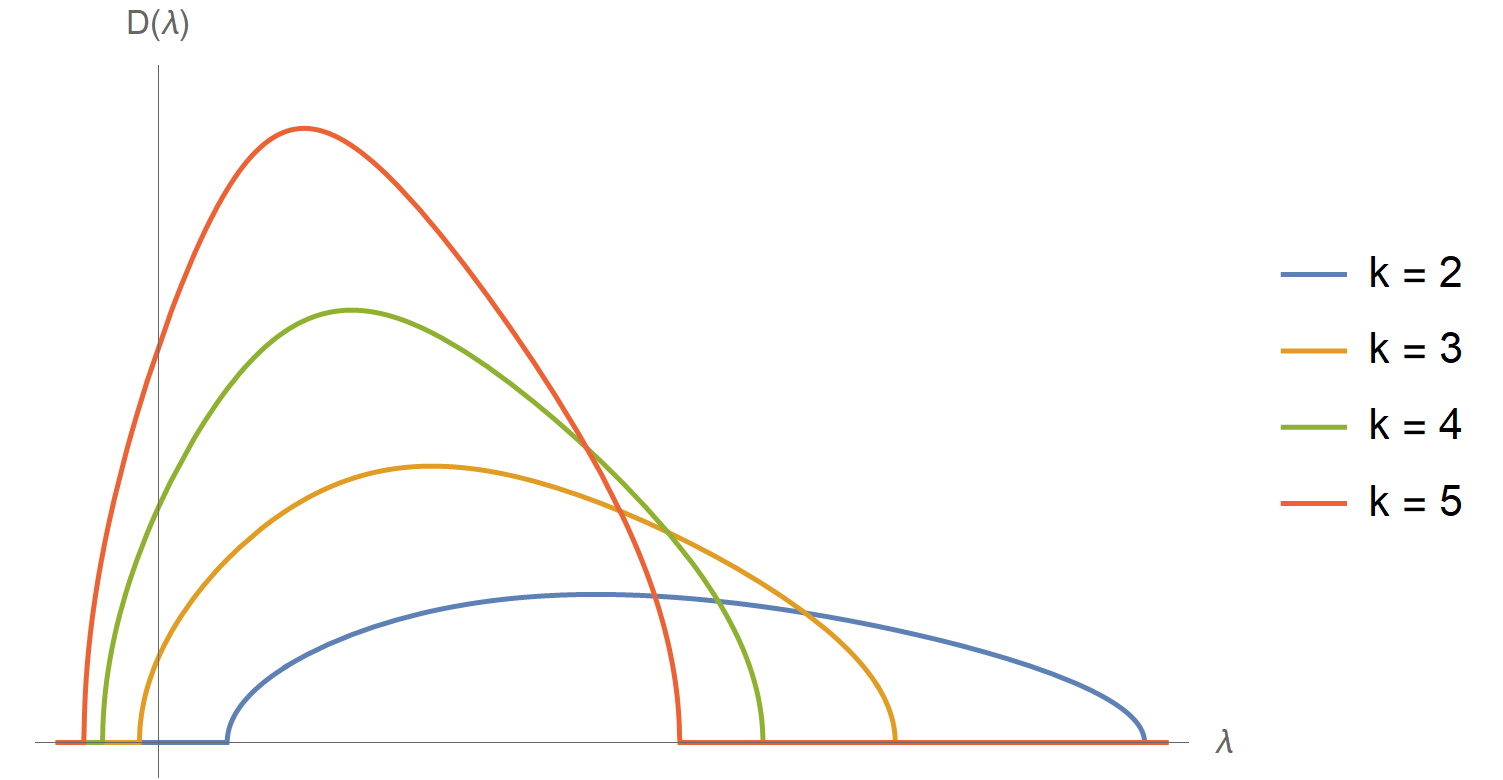}
\includegraphics[width=0.5\textwidth]{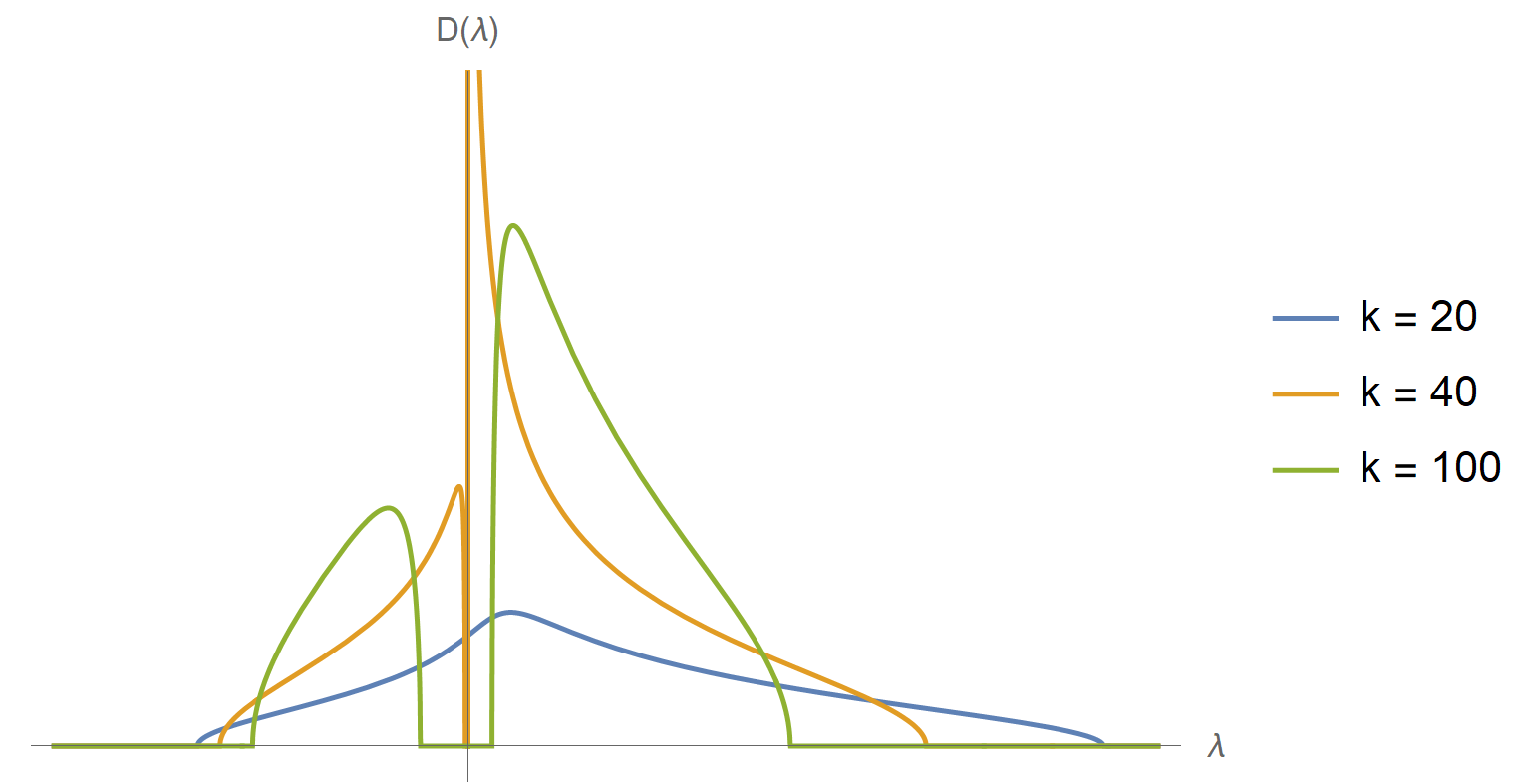}
}
\caption{Negativity spectrum in the microcanonical ensemble calculated from the cubic resolvent equation~\eqref{eq:MCsimple}. We fix $e^{\bfS} = 10$, $k_2 = 2$, and move in a horizontal line in the phase diagram by tuning $k$. The disconnected-to-pairwise transition (left) occurs at $k = e^{\bfS} = 10$, though we only plot to $k=5$ for visual clarity, as the qualitative behavior is the same. The pairwise-to-cyclic transition (right) occurs at $k = k_2^2 e^{\bfS} = 40$.}
\label{fig:MC}
\end{figure}

In the planar regime $k_1 \gg k_2 e^{-S_0}$, the resolvent equation \eqref{eq:resolvent_planar} becomes
\be
R^3 + \left( \frac{e^{\bfS}k_2^2 - k}{\lambda}\right)R^2 + e^{2\bfS}k k_2^2\left( \frac{1}{\lambda} - k \right)R + \frac{e^{2\bfS}k^3 k_2^2}{\lambda} = 0.
\label{eq:MCsimple}
\ee
This matches the resolvent equation derived in~\cite{Shapourian:2020mkc} for a random mixed state under appropriate rescaling of variables.\footnote{To match to the resolvent equation in~\cite{Shapourian:2020mkc}, define the rescaled variables $z = k_2 e^\bfS \lambda$ and $G = e^{-\bfS}R/k k_2$ so that \eqref{eq:MCsimple} becomes
\be
z G^3 + (\beta - 1) G^2 + (\alpha - z) G + 1 = 0,
\ee
where $\alpha = e^\bfS/k_1$ and $\beta = k_2 e^\bfS/k_1$. This cubic equation was earlier noted in the context of free probability theory \cite{banica2013asymptotic}.} As was shown there, a closed-form solution to this cubic equation for $R$ can be found, and leads to concrete results for the negativity spectrum and various negativity measures. The resolvent for the anti-planar regime $k_2 \gg k_1 e^{-S_0}$ can be obtained by $k_1 \leftrightarrow k_2$.

We plot the eigenvalue density in the microcanonical ensemble for various parameter values in Figure \ref{fig:MC}. The spectrum is approximately a Wigner semicircle distribution in the disconnected phase, continues to be connected in the pairwise phase, develops singularities at the pairwise-to-cyclic transition, and has two branches in the cyclic phase, where it is well approximated by the difference of two disjoint Marchenko-Pastur distributions.

\subsection{Canonical ensemble: disconnected-pairwise transition}\la{sec:dispair}

The transitions that involve the cyclic and anti-cyclic phases are complicated, as they involve a sum over diagrams with pieces connecting more than two asymptotic boundaries. Here, we will focus on the transition between the totally disconnected phase and the pairwise connected phase. The disconnected phase involves single-boundary diagrams, while the pairwise phase involves pairwise connected wormholes (plus a single disconnected piece for odd $n$).

The disconnected-pairwise transition happens within the large overlap $e^{-S_0} \ll k_1/k_2 \ll e^{S_0}$ between the planar regime and the anti-planar regime\footnote{In other words, we stay away from the two triple points on the phase diagram, as we would need to analyze transitions to the (anti-)cyclic phase there.}.  Therefore, the dominant geometries are those that are simultaneously planar and anti-planar.  As we show in Appendix~\re{app:dom}, these geometries are disjoint, non-crossing unions of single-boundary disks and pairwise connected wormholes.  This result is the content of Lemma~\re{lmdispair} in Appendix~\re{app:dom}.

Intuitively, these dominant geometries interpolate between the disconnected and pairwise geometries.  Geometries with pieces connecting more than two asymptotic boundaries are parametrically suppressed. As such, the resolvent equation~\er{eq:simple_SD_planar} truncates at the quadratic order:
\be
\lambda R = k + \frac{R}{k} + \frac{Z_2 R^2}{(kZ_1)^2}.
\ee
This quadratic equation can be solved analytically giving the resolvent and eigenvalue density as
\ba
R(\lambda) &= \frac{2k}{A^2} \left( \lambda - \frac{1}{k} - \sqrt{\lambda - \frac{1}{k} + A} \; \sqrt{\lambda - \frac{1}{k} - A} \right), \nonumber \\
D(\lambda) &= \frac{2k}{\pi A^2} \sqrt{A^2 - \left(\lambda - \frac{1}{k}\right)^2},
\label{eq:easyD}
\ea
where $A^2 \equiv 4Z_2/(kZ_1^2)$. Thus the eigenvalue density is a Wigner semicircle distribution supported on $\lambda \in \[-A+\fr{1}{k}, A+\fr{1}{k}\]$. For $k \leq 1/A$, $D(\lambda)$ only has support on $\lambda \geq 0$ and the negativity vanishes; for $k > 1/A$, $D(\lambda)$ has support on $\lambda < 0$ and we find the negativity is non-vanishing. The phase transition thus occurs at $k = 1/A \sim e^{S_0}$, as expected from the schematic analysis in Section \ref{sec:domsaddles}.

\begin{figure}
\centering
\includegraphics[width=.6\textwidth]{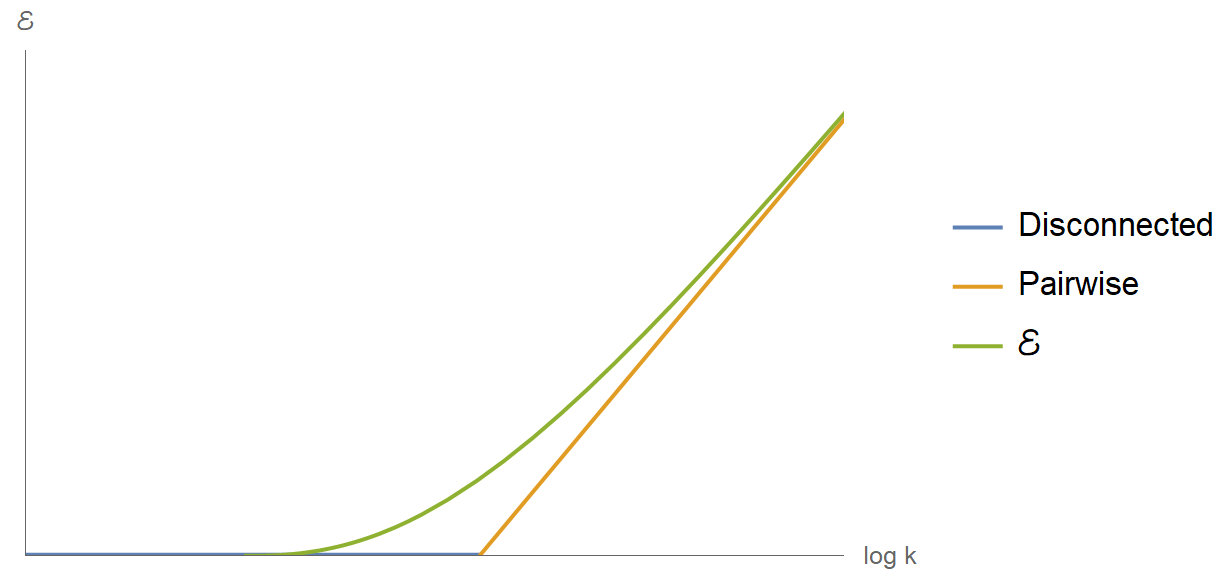}
\caption{Logarithmic negativity near the disconnected-to-pairwise transition with $Z_2/Z_1^2 = e^{-5}$, along with the naive answers in the dominant regions: $\mathcal{E}_{\textrm{disconnected}} = 0$ and $\mathcal{E}_{\textrm{pairwise}} = \frac{1}{2}\left( \log k + \log \frac{Z_2}{Z_1^2}\right) + \log \frac{8}{3 \pi}$, where $\log \frac{8}{3 \pi}$ is an $\mathcal{O}(1)$ term arising from the analytic continuation of the Catalan number in $\mathcal{N}^{\textrm{(even)}}_{2m}$ to $m = 1/2$ \cite{Dong:2021clv}.}
\label{fig:easyTrans}
\end{figure}

We can write an explicit expression for the logarithmic negativity using \eqref{eq:logneg} and \eqref{eq:easyD}:
\ba
\mathcal{E} &= \log \int_{-\infty}^\infty d \lambda D(\lambda) \abs{\lambda} \\
&= \log  \left[ \frac{2}{3 \pi} \left( \frac{\sqrt{A^2k^2-1}\left( 2A^2 k^2 + 1 \right)}{A^2 k^2} + 3 \csc^{-1} \left( A k \right)  \right) \right] \Theta \left( k - \frac{1}{A}\right).
\ea
We plot this result in Figure \ref{fig:easyTrans}, along with the naive answers for logarithmic negativity in the dominant phases.

How large is the correction at the transition? It is easy to verify that it is $\cO(1)$. There are no enhanced corrections here, as we are working in a regime where higher order terms in the Schwinger-Dyson equation are parametrically suppressed. The logarithmic negativity, along with all other negativity measures, never receives contributions from geometries containing pieces with $Z_{n>2}$, so corrections are $\mathcal{O}(1)$.

\subsection{Canonical ensemble: cyclic-pairwise transition}

In this subsection we will study the richer phase transition between the pairwise phase and the cyclic phase.\footnote{The pairwise to anti-cyclic transition follows from the calculation in this subsection by exchanging $k_1 \leftrightarrow k_2$.} Our computation is inspired by that of Appendix F of \cite{Penington:2019kki}.

First, let us define some useful values of $s$. In the semiclassical $\beta \ll 1$ and large brane mass $\mu \gg 1/\beta$ limits, we have
\be
\rho(s) y(s)^n \sim \frac{s}{2\pi^2} y(0)^n e^{2 \pi s - n \beta s^2/2}.
\label{eq:rhosemiclassical}
\ee
This means that the integral that defines $Z_n$ in \eqref{eq:Zn} can be well approximated by the saddle point located at
\ba
s^{(n)} &= \frac{2\pi}{n \beta} + \mathcal{O}(1) \nonumber \\
\Rightarrow \log Z_n &\approx S_0 + \frac{2 \pi^2}{n \beta} + \mathcal{O}(\log \beta).
\label{eq:Znapprox}
\ea
Throughout, we take our parameters $S_0$, $k_1$, and $k_2$ to be large before taking the semiclassical limit, such that e.g.\ $\log Z_n \approx S_0$ as $S_0 \gg 1/\beta$. We will also need to define $s_k$, the value of $s$ for which
\be
k = k_2^2 e^{S_0} \int_0^{s_k} ds \rho(s).
\label{eq:skdef}
\ee
We can approximate $s_k$ by
\be
s_k \approx \frac{1}{2\pi} \log \left( \frac{k}{k_2^2} \right) - S_0 + \mathcal{O}(1).
\label{eq:skapprox}
\ee
Note that here we are considering the values of $k$ and $k_2$ at transition, so the particular values of $s_k$ we are interested in will depend on the details of the negativity measure we are computing. In our schematic analysis where $Z_n \sim e^{S_0}$, we derived the location of the phase transition between the cyclic and pairwise phase and found that it was independent of $n$. However, taking into account dependence on $\beta$, the $Z_n$'s are distinct for different $n$, which leads to $n$-dependent transition points. The R\'enyi negativities in the cyclic and pairwise phases are given by the contributions of the dominant geometries in each phase (see Table \ref{table}), and coincide at transition. Equating their contributions at the transition, we find, up to factors $\mathcal{O}(1)$ in $\beta$,
\ba
\mathcal{N}_{2m}^{\textrm{(even)}} &= \frac{Z_{2m}}{k_2^{2m-2}Z_1^{2m}} = \frac{Z_2^m}{k^{m-1}Z_1^{2m}}, \nonumber \\
\mathcal{N}_{2m-1}^{\textrm{(odd)}} &= \frac{Z_{2m-1}}{k_2^{2m-2}Z_1^{2m-1}} = \frac{Z_2^{m-1}}{k^{m-1} Z_1^{2m-2}}.
\ea
In terms of the approximation \eqref{eq:Znapprox}, we can solve for $\log \(k/k_2^2\)$ at transition to obtain
\ba
&\textrm{Even:} && \log \(\frac{k}{k_2^2}\) = \log \left( \frac{Z_2^m}{Z_{2m}} \right)^{\frac{1}{m-1}} = S_0 + \left( 1+ \frac{1}{m} \right) \frac{\pi^2}{\beta} + \mathcal{O}\(\log \beta\) \nonumber \\
&\textrm{Odd:} && \log \(\frac{k}{k_2^2}\) = \log \left( \frac{Z_2^{m-1}Z_1}{Z_{2m-1}} \right)^{\frac{1}{m-1}} = S_0 + \left( 1+ \frac{4}{2m-1} \right) \frac{\pi^2}{\beta} + \mathcal{O}\(\log \beta\).
\ea
From this, we can solve for $s_k$ at the transition using \eqref{eq:skapprox} to obtain
\ba
s_k^{\textrm{($n$, even)}} &\approx \frac{\pi}{2 \beta} \left(1+\frac{2}{n} \right)  + \mathcal{O}\(\log \beta\) \nn\\
s_k^{\textrm{($n$, odd)}} &\approx \frac{\pi}{2 \beta} \left(1+\frac{4}{n} \right)+ \mathcal{O}\(\log \beta\).
\label{eq:skn}
\ea
As expected, the transition point depends on $n$. In particular, it is $\mathcal{O}(1/\beta)$ at leading order and bounded below as a function of $n$.

For this phase transition, we will fix $k$ and tune $k_2$. In the phase diagram, this corresponds to moving along a line between the upper left corner and the lower right corner. We need to consider diagrams with pieces made of an arbitrary number of boundaries, and we can restrict ourselves to planar diagrams, as anti-planar diagrams are suppressed by factors of $k_2/k_1$ relative to their planar counterparts. The resolvent equation is again \eqref{eq:Rint2}:
\ba
\lambda R = k + k_2^2 e^{S_0} \int_0^{\infty} ds \rho(s) \frac{w(s)R(k+w(s)R)}{k^2 k_2^2 - w(s)^2 R^2}.
\ea
We are going to split this integral at some transition $s_t$ such that
\be
\lambda R = k + k_2^2 e^{S_0} \int_0^{s_t} ds \rho(s) \frac{w(s)R(k + w(s)R)}{k^2k_2^2 - w(s)^2R^2} + k_2^2 e^{S_0} \int_{s_t}^\infty ds \rho(s) \frac{w(s)R(k + w(s)R)}{k^2k_2^2 - w(s)^2R^2}.
\label{eq:mideq}
\ee
We rewrite this simple step to emphasize that no approximations have been used yet. We are now going to use a set of three assumptions:
\begin{enumerate}
    \item $w(s_t)R \ll k k_2$.
    \item $k_2^2 e^{S_0} \int_0^{s_t} ds \rho(s) \frac{w(s)R(k + w(s)R)}{k^2k_2^2 - w(s)^2R^2} \ll k$.
    \item $s_t = s_k - \kappa$, where $\kappa$ is $\mathcal{O}(1)$ but large.
\end{enumerate}
These assumptions are justified in detail in Appendix \ref{app:trans}, where we show that the resulting simplifications to the resolvent equation give a self-consistent treatment of the problem. For now we will take these as facts and proceed. The first approximation allows us to simplify the final term in \eqref{eq:mideq} such that
\be
\lambda R \approx k + k_2^2 e^{S_0} \int_0^{s_t} ds \rho(s) \frac{w(s)R(k + w(s)R)}{k^2k_2^2 - w(s)^2R^2} + \frac{e^{S_0}}{k} \int_{s_t}^\infty ds \rho(s) w(s) R,
\ee
where we have dropped an $R^2$ term from the last integral because it can be shown to be much smaller than the leading term $k$ using Assumptions 1 and 3.
We define the coefficient of $R$ in the last term to be the constant $\lambda_0$, given by
\be
\lambda_0 \equiv \frac{e^{S_0}}{k} \int_{s_t}^\infty ds \rho(s) w(s).
\ee
We can now write the resolvent equation as 
\be
(\lambda - \lambda_0) R \approx k + k_2^2 e^{S_0} \int_0^{s_t} ds \rho(s) \frac{w(s)R(k + w(s)R)}{k^2k_2^2 - w(s)^2R^2} .
\label{eq:mideq2}
\ee
Now we turn to Assumption~2, which allows us to treat the second term above as a perturbation to the zeroth order solution
\be
R \approx R_0 = \frac{k}{\lambda - \lambda_0}.
\ee
Plugging this solution back into \eqref{eq:mideq2}, we obtain the first order iterated solution
\ba
R_1 &\approx \frac{k}{\lambda - \lambda_0} + \frac{k_2^2e^{S_0}}{\lambda - \lambda_0} \int_0^{s_t} ds \rho(s) \frac{w(s)R_0(k+w(s)R_0)}{k^2 k_2^2 - w(s)^2R_0^2} \nonumber \\
&\approx \frac{k}{\lambda - \lambda_0} + e^{S_0} \int_0^{s_t} ds \frac{1}{\lambda - \lambda_0} \frac{\rho(s)w(s)(\lambda - \lambda_0 + w(s))}{(\lambda - \lambda_0)^2 - \left( \frac{w(s)}{k_2}\right)^2}.
\ea
Now we can find the discontinuity in this expression and extract $D(\lambda)$. There are three contributions to the spectrum: a simple pole at $\lambda = \lambda_0$, and a pair of branch cuts given by the poles at $\lambda = \lambda_0 \pm w(s)/k_2$ in the integrand. We obtain
\bm
D(\lambda) = \# \delta(\lambda - \lambda_0) + e^{S_0} \int_0^{s_t} ds \rho(s) \left[ \frac{k_2(k_2+1)}{2} \delta\left(\lambda - \lambda_0 - \frac{w(s)}{k_2}\right) \right.\\
\left.+ \frac{k_2(k_2-1)}{2} \delta \left(\lambda - \lambda_0 + \frac{w(s)}{k_2} \right) \right].
\label{eq:Dapprox1}
\em
\begin{figure}
    \centering
    \includegraphics[width=\textwidth]{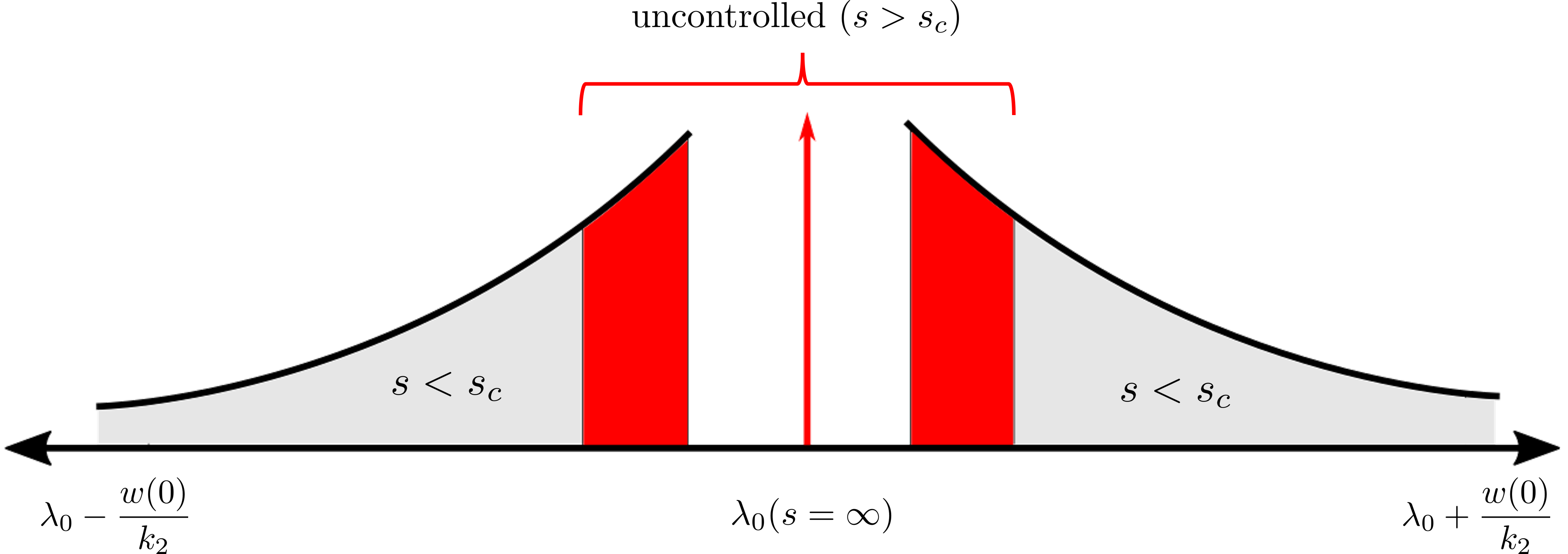}
    \caption{A rough sketch (not to scale) of the negativity spectrum near the cyclic-to-pairwise transition. }
    \label{fig:specapprox}
\end{figure}Let us pause for a second to unpack this equation. The spectrum consists of a delta function located at $\lambda_0$ from the simple pole and two regions of nonzero eigenvalue density from the integrated delta functions. We plot a sketch of this eigenvalue density in Figure \ref{fig:specapprox}. There are two distinct regions with nonzero eigenvalue density, similar to the spectrum in the microcanonical ensemble. One point we emphasize in Appendix \ref{app:trans} is the presence of a ``controlled'' region $0 < s < s_c$ in which our assumptions hold and an ``uncontrolled'' region $s > s_c$ where we claim ignorance about the spectrum. In terms of the eigenvalues, this corresponds to an ignorance in the spectrum for a region
\be
\lambda \in \left[ \lambda_0 - \frac{w(s_c)}{k_2}, \lambda_0 + \frac{w(s_c)}{k_2} \right].
\ee
We show that our ignorance about the uncontrolled region leads to at most $\cO(1)$ multiplicative corrections to the R\'enyi negativities, or $\cO(1)$ additive corrections to $\mathcal{E}$, $S^{T_2}$, and $S^{T_2(2)}$, due to the constraint that the total number of eigenvalues must be $k$.

Clearly, $\lambda_0$ lies within this uncontrolled region, so we should not take seriously the presence of the delta function at $\lambda_0$. In fact, we will now show that this delta function vanishes if we extend the upper limit of the integral from $s_t$ to $s_k$ in~\er{eq:Dapprox1} (which only affects the uncontrolled region and therefore causes a small error). Our density matrix has a total of $k$ eigenvalues and is unit normalized, which translates into conditions on the zeroth and first moments of $D(\lambda)$, namely
\be
\int_{-\infty}^\infty d\lambda D(\lambda) = k, \; \; \; \; \int_{-\infty}^\infty d\lambda D(\lambda) \lambda = 1.
\ee
We see that these conditions are satisfied by \eqref{eq:Dapprox1} if we replace $s_t$ by $s_k$, compensating for the fact that $s_t = s_k - \kappa$ by sending the coefficient of the $\delta (\lambda - \lambda_0)$ piece to zero. We conclude that a good approximation for the spectrum of the partially transposed density matrix at transition is given by 
\be
D(\lambda) =  e^{S_0} \int_0^{s_k} ds \rho(s) \left[ \frac{k_2(k_2+1)}{2} \delta\left(\lambda - \lambda_0 - \frac{w(s)}{k_2}\right) + \frac{k_2(k_2-1)}{2} \delta \left(\lambda - \lambda_0 + \frac{w(s)}{k_2} \right) \right]
\label{eq:NegDApprox}
\ee
In order to simplify our calculations for negativity, we would like that the delta function branch cuts were purely positive or negative, which is equivalent to the condition $\lambda_0 < w(s_k)/k_2$. By definition, $\lambda_0$ is bounded above by $1/k$, and at transition we have $k/k_2^2 = e^{S_0 + 2 \pi s_k}$. In the semiclassical limit, we can use the approximation
\be
w(s_k) \approx \frac{e^{-\beta s_k^2/2}}{Z_1} \approx e^{-S_0 -\beta s_k^2/2 - 2\pi^2/\beta}.
\ee
Our condition on the branch cuts becomes
\be
\frac{kw(s_k)}{k_2} > 1 \qu \Rightarrow \qu k_2 e^{S_0+2\pi s_k} w(s_k) \approx k_2 e^{2 \pi s_k - \beta s_k^2/2 - 2\pi^2/\beta} \sim k_2e^{C/\beta} \gg 1
\label{eq:wapprox2}
\ee
as $s_k \sim 1/\beta$ for a generic $n$. This is satisfied even if the unknown order one constant is negative under our previous assumption that we take our counting parameters to be large before taking small $\beta$, such that $\log k_2 \gg \frac{1}{\beta}$.

Now that we have the spectrum at transition, we are ready to calculate the corrections to any negativity measure we want! Let us start with the logarithmic negativity, which has a transition located at
\be
s_k^{(1, \textrm{even})} = \frac{3 \pi}{2 \beta}.
\ee
We find
\ba
\mathcal{E} &= \log \int_{-\infty}^{\infty} d\lambda D(\lambda) \abs{\lambda} \nonumber \\
&= \log e^{S_0} \int_0^{s_k} ds \rho(s) \left[ \frac{k_2(k_2+1)}{2} \left( \lambda_0 + \frac{w(s)}{k_2} \right) + \frac{k_2(k_2-1)}{2} \left( \frac{w(s)}{k_2} - \lambda_0 \right) \right] \nonumber \\
&= \log \left( \frac{1}{k_2} + k_2 e^{S_0} \int_0^{s_k} ds \rho(s)w(s) \right).
\ea
As $s_k < s^{(1)}$ for logarithmic negativity, we have approximated $\lambda_0 \approx 1/k$, and the final integral is well approximated by its maximum value $\rho(s_k) w(s_k)$. The logarithmic negativity is then
\ba
\mathcal{E} &\approx \log \left( \frac{1}{k_2} + k_2 e^{S_0} \rho(s_k)w(s_k) \right) \nonumber\\
&\approx \log \left( \frac{1}{k_2} + \frac{k w(s_k)}{k_2}\right) \nonumber \\
&\approx \log k_2 - \frac{\pi^2}{8 \beta}. 
\ea
In the second line we used our previous approximation \eqref{eq:skapprox} for $s_k$, and in the third line we used \eqref{eq:wapprox2}. As we see, the logarithmic negativity experiences an $\mathcal{O}(1/\beta)$ correction to the naive answer $\mathcal{E} = \log k_2$.\\

Where do we expect $\mathcal{O}(1/\sqrt{\beta})$ corrections? In the case of even R\'enyi negativities, this happens at
\be
s_k^{(n,\textrm{even})} = s^{(n)} \Rightarrow n = 2, \; \;  s^{(2)} = \frac{\pi}{\beta}.
\ee
We can check this explicitly for a negativity measure descending from the even analytic continuation. The simplest such measure, the R\'enyi-2 negativity $\mathcal{N}_2^{(\textrm{even})}$, is related to the second R\'enyi entropy $S_2$ by
\be
\mathcal{N}_2^{(\textrm{even})} = e^{-S_2}
\ee
and therefore comes with $\mathcal{O}(1)$ corrections. We instead turn to the refined R\'enyi-2 negativity $S^{T_2(2)}$, defined in \eqref{eq:secondrefined}. We can read off the naive answer for $S^{T_2(2)}$ from Table \ref{table}, again using the approximation \eqref{eq:Znapprox}. We find
\be
S^{T_2(2)} \approx 2 \log k_2 + S_0 + \frac{2\pi^2}{\beta} .
\ee
To compute $S^{T_2(2)}$, we first need to compute $\sum_i \lambda_i^2 = \mathcal{N}_2^{(\textrm{even})}$. At  transition, the naive answer is given by Table \ref{table}, where $\mathcal{N}_2^{(\textrm{even})} = Z_2/Z_1^2$. However, using \eqref{eq:NegDApprox}, we find
\ba
\mathcal{N}_2^{(\textrm{even})} &= \int_{-\infty}^\infty d \lambda D(\lambda) \lambda^2 \nonumber \\
&= e^{S_0} \int_0^{s_k} ds \rho(s) \left( k_2^2 \lambda_0^2 + 2 \lambda_0 w(s) + w(s)^2 \right)
\ea
Again, $s_k < s^{(1)}$, so $\lambda_0 \approx 1/k$. This integral gives
\ba
\mathcal{N}_2^{(\textrm{even})} &= \frac{1}{k} + \frac{2e^{S_0}\rho(s_k)w(s_k)}{k} + \frac{Z_2}{2Z_1^2} \nonumber \\
&\approx \frac{1}{k} + \frac{2e^{-\pi^2/2\beta}}{k} + \frac{Z_2}{2Z_1^2}.
\ea
In the limit $k \gg e^{S_0}$, we can safely ignore the first two terms, and the R\'enyi-2 negativity becomes
\be
\mathcal{N}_2^{(\textrm{even})} \approx \frac{Z_2}{2Z_1^2} \approx \frac{\sqrt{\pi}}{4}\beta^{3/2} e^{-S_0 - 3 \pi^2 / \beta}.
\ee
The factor of $1/2$ out front may seem like a problem, as we do not reproduce the naive answer for $\mathcal{N}^{\textrm{(even)}}_2$. However, as the refined R\'enyi negativities are functions of $\log \mathcal{N}_n$, this factor will only contribute an $\mathcal{O}(1)$ difference from the true answer for $S^{T_2(2)}$, and we are safe in using this approximation. The refined R\'enyi-2 negativity is therefore given by
\ba
S^{T_2(2)} = &-\int_{-\infty}^\infty d\lambda D(\lambda) \frac{\lambda^2}{\mathcal{N}_2^{(\textrm{even})}} \log \frac{\lambda^2}{\mathcal{N}_2^{(\textrm{even})}} \nonumber \\
= &-\frac{e^{S_0}}{\mathcal{N}_2^{(\textrm{even})}} \int_0^{s_k} ds \rho(s)  \Biggl( \frac{k_2(k_2+1)}{2}\left(\lambda_0+\frac{w(s)}{k_2} \right)^2 \log \left(\lambda_0+\frac{w(s)}{k_2} \right)^2  \nonumber \\   
&+ \frac{k_2(k_2-1)}{2}\left(\lambda_0-\frac{w(s)}{k_2} \right)^2 \log \left(\lambda_0-\frac{w(s)}{k_2} \right)^2 \Biggr) \nonumber \\
&+ \log \mathcal{N}_2^{(\textrm{even})}
\ea
Again, as $\lambda_0 \approx 1/k$, the dominant contribution to this integral will come from the $w(s)^2$ term, which is where the enhanced transition should come in. Previously, using \eqref{eq:rhosemiclassical}, we showed that an integral of the form $\rho(s)w(s)^n$ can be approximated by a sharply peaked Gaussian with mean $s^{(n)}$ and standard deviation $1/\sqrt{n\beta}$, up to normalization. We use these simplifications to obtain
\ba
& -\frac{e^{S_0}}{\mathcal{N}_2^{(\textrm{even})}} \int_0^{s_k} ds \rho(s) w(s)^2 \log \frac{w(s)^2}{k_2^2} \\
&= -\frac{e^{S_0}}{\mathcal{N}_2^{(\textrm{even})}} \int_0^{s_k} ds \rho(s) w(s)^2 \left( \log \frac{w(s^{(2)})^2}{k_2^2} + \log \frac{w(s)^2}{w(s^{(2)})^2}\right) \nonumber \\
&=  - \log \left(\frac{w(s^{(2)})^2}{k_2^2} \right) - \frac{e^{S_0}}{\mathcal{N}_2^{(\textrm{even})}} \int_0^{s_k} ds  \left( \rho(s) w(s)^2 \log \frac{w(s)^2}{w(s^{(2)})^2} \right) \nonumber \\
&\approx - \log \left(\frac{w(s^{(2)})^2}{k_2^2} \right) - 2 \sqrt{\frac{\beta}{ \pi}} \int_0^{s^{(2)}} ds e^{-\beta\(s-s^{(2)}\)^2} \beta \( \(s^{(2)}\)^2 - s^2 \) \nonumber \\
&\approx 2 \log k_2 + 2S_0 + \frac{5 \pi^2}{ \beta} - \sqrt{\frac{ 4 \pi}{\beta}} + \mathcal{O}(1)
\ea
Our final expression for $S^{T_2(2)}$ is therefore
\be
S^{T_2(2)} = 2 \log k_2 + S_0 + \frac{2\pi^2}{\beta} -  \sqrt{\frac{4\pi}{\beta}}
\ee
confirming that there is an $\mathcal{O}(1/\sqrt{\beta})$ correction at transition. 

The fact that the refined R\'enyi-2 negativity experiences this particular correction is not surprising due to its close connection to von Neumann entropies.  It is known that von Neumann entropies receive $\mathcal{O}(1/\sqrt{\beta})$ or $\mathcal{O}(1/\sqrt{G_N})$ corrections at the Page transition~\cite{Penington:2019kki,Dong:2020iod,Marolf:2020vsi}, which can be explained using a diagonal approximation with respect to a basis of fixed-area states~\cite{Dong:2020iod,Marolf:2020vsi,Akers:2018fow}.  It was shown that the refined R\'enyi-2 negativity can be written in holography as the sum of the von Neumann entropies of $R_1$ and $R_2$ in the state $\rho_{R_1 R_2}^2$ (once properly normalized)~\cite{Dong:2021clv}.\footnote{This can be understood in terms of two cosmic branes homologous to $R_1$ and $R_2$, respectively, in the gravity dual of the even R\'enyi negativity.  These cosmic branes arise from a $\bZ_{n/2}$ quotient and therefore become tensionless in the $n\to 2$ limit, which is similar to the case of the von Neumann entropy.}  Therefore, the $\mathcal{O}(1/\sqrt{\beta})$ correction that we find in the refined R\'enyi-2 negativity can similarly be explained using the diagonal approximation.

As we show in Appendix~\ref{app:renyi}, the R\'enyi entropy $S_n$ with $n < 1$ experiences $\mathcal{O}(1/\beta)$ corrections in the model of \cite{Penington:2019kki}, as there too we are computing an entanglement measure with $s_k^{(n)} < s^{(n)}$. In other words, the R\'enyi index of both the logarithmic negativity and the R\'enyi entropy with $n < 1$ is below some ``critical'' R\'enyi index at which there exist $\mathcal{O}(1/\sqrt{G})$ corrections. 

For measures descending from odd R\'enyi negativity, we might not expect $\mathcal{O}(1/\sqrt{\beta})$ corrections, as we never have $s_k^{(n,\textrm{odd})} = s^{(n)}$. However, we may still expect some enhanced corrections for some negativity measures in this case. The partially transposed entropy $S^{T_2}$ is one such measure. As $s_k^{(1,\textrm{odd})} = 5\pi/2\beta$, the naive answer for $S^{T_2}$ is given by 
\be
S^{T_2} = \log k_2 + S_0 + \frac{4 \pi^2}{\beta}.
\label{eq:naiveST}
\ee
Our approximation gives
\ba
S^{T_2} &= - \int_{-\infty}^\infty d\lambda D(\lambda) \lambda \log \abs{\lambda} \nonumber \\
= &-e^{S_0} \int_0^{s_k} ds \rho(s)  \Biggl( \frac{k_2(k_2+1)}{2}\left(\lambda_0+\frac{w(s)}{k_2} \right) \log \left(\lambda_0+\frac{w(s)}{k_2} \right)  \nonumber \\   
&+ \frac{k_2(k_2-1)}{2}\left(\lambda_0-\frac{w(s)}{k_2} \right) \log \left(\frac{w(s)}{k_2} - \lambda_0 \right) \Biggr).
\label{eq:ST1}
\ea
There is however a subtlety here. The dominant contribution to this integral no longer comes solely from the $w(s)$ term. This can be seen by expanding \eqref{eq:ST1} using
\be
\log \left( \frac{w(s)}{k_2} \pm \lambda_0 \right) \approx \log \frac{w(s)}{k_2} \pm \frac{\lambda_0 k_2}{w(s)}.
\ee
From this we obtain
\ba
S^{T_2} &= -e^{S_0} \int_0^{s_k} ds \rho(s) \left( k_2^2 \left( \lambda_0 + \lambda_0 \log \frac{w(s)}{k_2}  \right) + k_2 \left( \frac{\lambda_0^2 k_2}{w(s)} + \frac{w(s)}{k_2} \log \frac{w(s)}{k_2} \right) \right) \nonumber \\
&\approx -e^{S_0} \int_0^{s_k} ds \rho(s) \left( k_2^2 \lambda_0 + w(s) \right) \log \frac{w(s)}{k_2}.
\ea
Treating the $\log \frac{w(s)}{k_2}$ as negligible compared to the exponential $\rho(s)$, the two terms in the integrand are of the same order, so we should keep them both. If we look at the naive transition point $s_k = 5\pi/2\beta > s^{(1)}$, we find
\ba
S^{T_2} &\approx -e^{S_0} \int_0^{s_k} \rho(s) w(s) \log \frac{w(s)}{k_2} \nonumber \\
&\approx -\log \frac{w(s_k)}{k_2} \nonumber \\
&\approx \log k_2 + S_0 + \frac{4 \pi^2}{\beta},
\ea
and we would conclude that the correction is $\mathcal{O}(1)$. However, if we were to find the largest correction to this quantity, we would look not at the naive transition, but at the point where we might find $\mathcal{O}(1/\sqrt{\beta})$ corrections, at $s_k = s^{(1)}$. At this point $\lambda_0 = 1/2k$ and we capture half of the Gaussian $\rho(s) w(s)$, so we have
\ba
S^{T_2} &\approx -\frac{1}{2} \log \frac{w(s^{(1)})}{k_2} - \sqrt{\frac{\beta}{2 \pi}} \int_0^{s^{(1)}} e^{-\beta (s - s^{(1)})^2/2} \frac{\beta}{2} \log \frac{w(s)}{k_2} \nonumber \\
&\approx  -\log \frac{w(s^{(1)})}{k_2} - \sqrt{\frac{\beta}{2 \pi}} \int_0^{s^{(1)}} e^{-\beta (s - s^{(1)})^2/2} \frac{\beta}{2} \left( \left( s^{(1)}\right)^2 - s^2\right) \nonumber \\
&\approx \log k_2 + S_0 + \frac{4\pi^2}{\beta} - \sqrt{\frac{2\pi}{\beta}}.
\ea
This looks the same as the naive answer with a $\mathcal{O}(1/\sqrt{\beta})$ correction. However, as we are working at fixed $k$, we should really be writing everything in terms of $k$ and $e^{S_0}$ using \eqref{eq:skapprox}, in which case our naive and corrected $S^{T_2}$'s are
\begin{alignat}{3}
&\textrm{Naive:} \qqu && S^{T_2} = \frac{1}{2} \left( \log k - S_0 \right) + \frac{3 \pi^2}{2 \beta} \nonumber \\
&\textrm{Corrected:} \qqu && S^{T_2} = \frac{1}{2} \left( \log k - S_0 \right) + \frac{2 \pi^2}{\beta},
\end{alignat}
and we find an $\mathcal{O}(1/\beta)$ correction.

\section{Topological model with EOW branes}\label{sec:top}
Having studied a toy model of an evaporating black hole in JT gravity, we will now consider entanglement negativity in the context of the topological model of Marolf-Maxfield~\cite{Marolf:2020xie}, including dynamical end-of-the-world (EOW) branes. This is a theory of topological two-dimensional gravity in which spacetimes are two-dimensional manifolds endowed only with orientation. In contrast to the JT model of Section \ref{sec:jt}, there are no metric or dilaton degrees of freedom and EOW brane boundaries can be generated dynamically.

The action for the topological model is given by
\be
S_\text{top} = -S_0 \chi(M) - S_\partial |\partial M|,
\ee
where $S_0$ is some arbitrary parameter, $\chi(M)$ is the Euler characteristic of the (possibly disconnected) manifold $M$, and $|\partial M|$ counts the number of boundaries. $S_\partial|\partial M|$ is a non-local term that we put in by hand to ensure reflection positivity. As shown in~\cite{Marolf:2020xie}, the simplest choice which results in reflection positivity is $S_\partial = S_0$.\footnote{Other valid choices are $S_\partial = S_0 + \log m$ for any positive integer $m$ or $S_\partial > S_0 + \log k$. Any of these choices give a discrete spectrum for the operator $\hat Z$.} The action then becomes
\be
S_\text{top} = -S_0 \td{\chi}
\ee
where $\tilde{\chi} = 2-2g$ for any manifold, with or without boundary. To make contact with black hole evaporation, we can extend the model to include EOW branes, which can take one of $k$ ``flavors". Since we are interested in studying negativity, we will allow the branes to be labeled by a set of two of flavor indices $\{i,j\}$, where $i \in \{1, \dots, k_1\}$, $j\in \{1, \dots, k_2\}$, and such that $k = k_1k_2$. This is exactly analogous to our construction in Section \ref{sec:jt} and is a slight generalization of the model in~\cite{Marolf:2020xie}.

There are three distinct types of boundaries allowed by the theory. The first are circular asymptotically AdS boundaries denoted by $Z$. These boundaries are associated with an operator $\widehat Z$ which acts on the baby universe Hilbert space $\mathcal{H}_{\text{BU}}$ and creates a $Z$ boundary. Second, there are boundary conditions which we denote by $(\psi_{i_2j_2},\psi_{i_1j_1})$ composed of an oriented interval of asymptotically AdS boundaries with endpoints labeled by flavor indices $\{i_1,j_1\}$ and $\{i_2,j_2\}$. The diagram that describes this is the same as in \eqref{eq:braneBC}. These boundaries are associated with an operator $\widehat{ (\psi_{i_2j_2},\psi_{i_1j_1})}$ on $\mathcal{H}_{\text{BU}}$. Finally, there are circular EOW brane boundaries labeled by an arbitrary flavor index $\{i,j\}$, independent of all boundary conditions. These brane boundaries can be dynamically generated as additional boundaries when performing the gravitational path integral.

Let us consider the simplest quantity one can compute with the gravitational path integral of this theory, namely the partition function associated to a single connected component of spacetime with some number of asymptotic boundaries. The gravitational path integral demands that we sum over all such manifolds with arbitrary genus, weighted by $e^{-S_0 \td\chi}$. Additionally, since the EOW branes are dynamical, there is the possibility of the gravitational path integral generating an arbitrary number of closed brane boundaries, each of which contribute a factor of $k$ and are mutually indistinguishable. The partition function for a single connected component with some fixed number of asymptotic boundaries is therefore
\be
\lambda \equiv \sum_{g=0}^\infty \sum_{m=0}^\infty e^{(2-2g)S_0} \fr{k^m}{m!} = \fr{e^{2S_0}}{1-e^{-2S_0}} e^{k}.
\ee
More generally, one can consider amplitudes
\be \label{eq:top_amp}
\bigg \< Z^m (\psi_{i_1'j_1'},\psi_{i_1j_1}) \cdots (\psi_{i_n'j_n'},\psi_{i_nj_n}) \bigg \> \equiv \big \< \text{NB} \big\vert \widehat{Z}^m \widehat{(\psi_{i_1'j_1'},\psi_{i_1j_1})} \cdots \widehat{(\psi_{i_n'j_n'},\psi_{i_nj_n})}\big\vert \text{NB} \big \>
\ee
which are computed using the gravitational path integral by summing over all (possibly disconnected) manifolds with boundary conditions specified by $m$ circular boundaries $Z$ and $n$ oriented intervals $(\psi_{i'j'},\psi_{ij})$ with endpoints labeled by the corresponding flavor indices and connected to oriented brane boundaries labeled with matching flavors. The brackets in \eqref{eq:top_amp} can be interpreted the expectation value of the corresponding operators in the no-boundary state $\big\vert \text{NB} \big\> \in \mathcal{H}_\text{BU}$. In what follows, we will assume that the no-boundary state is unit normalized, $\big \< \text{NB} \big\vert \text{NB} \big \> = 1$.\footnote{In~\cite{Marolf:2020xie}, the no-boundary state has inner product $\big \< \text{NB} \big\vert \text{NB} \big \> = e^{\lambda}$ and represents the sum over arbitrary numbers of closed universes. This normalization enters as a universal prefactor in all amplitudes we compute, so we can choose to normalize it to one.}

Let us now proceed to the calculation of negativity in this model. In analogy to Section \ref{sec:jt}, we can define the (unnormalized) density matrix
\be
\rho = \sum_{i_1,i_2=1}^{k_1}\sum_{j_1,j_2=1}^{k_2} |i_1, j_1 \rangle \langle i_2, j_2|(\psi_{i_2j_2},\psi_{i_1j_1}),
\ee
which plays the role of the state of the Hawking radiation. We are interested in studying the R\'enyi negativities $\mathcal{N}_n = \Tr\[\(\rho^{T_2}\)^n\]$ of this density matrix. The most straightforward method is to use the moment generating function of $\rho^{i_1i_2}_{j_1j_2} \equiv (\psi_{i_2j_2},\psi_{i_1j_1})$, which one can show is given by
\be\label{eq:genfun}
\bigg\<\exp \bigg(\sum_{i_1,i_2=1}^{k_1}\sum_{j_1,j_2=1}^{k_2} t^{i_1i_2}_{j_1j_2}\rho^{i_1i_2}_{j_1j_2} \bigg)\bigg\> = e^{-\lambda} \exp\bigg[\lambda \det\(\mathbb{I}-t\)^{-1}\bigg],
\ee
where $t$ can be thought of as the $k \times k$ matrix with entries $t^{i_1i_2}_{j_1j_2}$ by treating $\{i,j\}$ as a single index of size $k$ and $\mathbb{I}$ is the $k\times k$ identity matrix . This is a slight generalization of the result derived in~\cite{Marolf:2020xie}. In principle, one can compute all moments of $\rho^{i_1i_2}_{j_1j_2}$, and hence all R\'enyi negativities, by taking appropriate partial derivatives of \eqref{eq:genfun} with respect to $t^{i_1i_2}_{j_1j_2}$. 

However, there is a shortcut that we will now describe. As shown in~\cite{Marolf:2020xie}, the spectrum of the operator $\widehat Z$ takes values in $\mathbb{N}$. One can derive the distribution for $\rho^{i_1i_2}_{j_1j_2}$ in a fixed $Z = d \in \mathbb{N}$ sector by Taylor expanding \eqref{eq:genfun} in $\lambda$:
\ba
&\bigg\<\exp \bigg(\sum_{i_1,i_2=1}^{k_1}\sum_{j_1,j_2=1}^{k_2} t^{i_1i_2}_{j_1j_2}\rho^{i_1i_2}_{j_1j_2} \bigg)\bigg\> = \sum_{d=0}^\infty p_d(\lambda) \bigg\< \exp\( \sum_{i_1,i_2=1}^{k_1}\sum_{j_1,j_2=1}^{k_2} t^{i_1i_2}_{j_1j_2}\rho^{i_1i_2}_{j_1j_2} \) \bigg\>_{Z=d} \nn \\
&\qqu \qqu \implies \bigg\< \exp \(\sum_{i_1,i_2=1}^{k_1}\sum_{j_1,j_2=1}^{k_2} t^{i_1i_2}_{j_1j_2}\rho^{i_1i_2}_{j_1j_2}\) \bigg\>_{Z=d} = \det\(I-t\)^{-d}, \label{eq:top_mgf_d}
\ea
where $p_d(\lambda) = e^{-\lambda}\fr{\lambda^d}{d!}$ is a Poisson distribution. We can recognize \eqref{eq:top_mgf_d} as the moment generating function for a Wishart distribution with $d$ degrees of freedom, and coincides with the distribution of a random mixed state and the microcanonical JT model in Section \ref{sec:MC}. 
Thus, we can immediately write down the R\'enyi negativities in a fixed $Z=d$ sector
\be\label{eq:fixedZ}
\Big\<\Tr\(\rho^{T_2}_R\)^n\Big\>_{Z=d} = \sum_{g\in S_n} d^{\chi(g)} k_1^{\chi(g^{-1}X)} k_2^{\chi(g^{-1}X^{-1})},
\ee
which matches the answer for the microcanonical ensemble \eqref{eq:microcanonical} with $d$ playing the role of $e^\bfS$. The results for the negativity spectrum obtained in Sections~\ref{sec:domsaddles} and~\ref{sec:MC} therefore apply with this replacement. However, since $d$ does not correspond to the partition function on some manifold, it is difficult to interpret the result in \eqref{eq:fixedZ} geometrically.

To obtain the R\'enyi negativities in the full theory, we simply sum over $d \in \mathbb{N}$ with Poisson weight $p_d\(\lambda\)$:
\ba
\Big\<\Tr\(\rho^{T_2}_R\)^n\Big\> &= \sum_{d=0}^\infty p_d(\lambda) \,\Big\<\Tr\(\rho^{T_2}_R\)^n\Big\>_{Z=d} \nonumber \\
&= \sum_{g\in S_n} B_{\chi(g)}(\lambda) k_1^{\chi(g^{-1}X)} k_2^{\chi(g^{-1}X^{-1})} \label{eq:top_neg}
\ea
where $B_m(\lambda) = e^{-x} \sum_{k=0}^\infty \fr{\lambda^k k^m}{m!}$ are the Bell polynomials, whose asymptotic behavior is $B_m(\lambda) \sim \lambda^m$ as $\lambda \to \infty$. We therefore find
\be\label{eq:top_neg_approx}
\Big\<\Tr\(\rho^{T_2}_R\)^n\Big\> \approx \sum_{g\in S_n} \lambda^{\chi(g)} k_1^{\chi(g^{-1}X)} k_2^{\chi(g^{-1}X^{-1})}, \qqu \lambda \gg 1.
\ee
This is once again equivalent to the microcanonical ensemble in \eqref{eq:microcanonical}, with $\lambda$ now playing the role of $e^\bfS$, and therefore we can obtain concrete results for the negativity spectrum. Since $\lambda$ is the gravitational partition function of a single connected component of spacetime, we can in fact find a geometric interpretation for the terms in \eqref{eq:top_neg}.

To understand the geometric origins of the terms in \eqref{eq:top_neg}, let us first look at the case $n=2$, which gives the purity
\ba\label{eq:purity}
\Big\< \Tr\rho^2 \Big\> = \lambda^2 k + \lambda k^2 + \lambda k.
\ea
The terms in \eqref{eq:purity} correspond to the following geometries:
\begin{figure}[H]
    \centering
    \vspace{4mm}
    \includegraphics[width=0.5\textwidth]{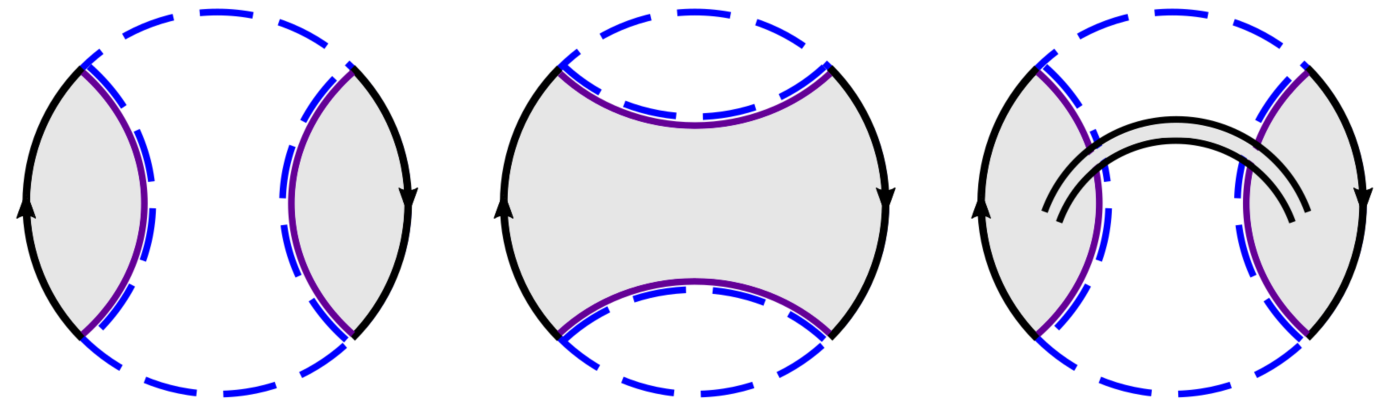}
    \label{fig:BU}
    \vspace{2mm}
\end{figure}
\noindent The first two diagrams are familiar: they are the disk and wormhole geometries, summed over genus and closed brane boundaries. The last diagram represents two disk geometries joined by an arbitrary number of wormholes; we thus call it a \textit{joining wormhole}.

More generally, the geometries which contribute at leading order in $\lambda$ in the R\'enyi negativities \eqref{eq:top_neg} are in one-to-one correspondence with elements of the permutation group. To be precise, each of these geometries is actually a disjoint union of disks, summed over genus and closed brane boundaries.\footnote{This is in contrast with the JT model where we identify only a single geometry, namely some disjoint union of disks with no handles, with each element of the permutation group. In that case, the sum over genus is highly suppressed by factors of $e^{-\bfS}$, and they are in fact as suppressed as geometries with joining wormholes. Furthermore, in the JT model closed brane boundaries can not be dynamically generated.} The subleading contributions in \eqref{eq:top_neg} can be identified with the same geometries but with arbitrary numbers of joining wormholes between connected components, and thus can not be mapped to elements of the permutation group.

It is clear that $\log \lambda$ plays the same role as $\bfS$ in the microcanonical JT model, namely it is the Bekenstein-Hawking entropy. In analogy to black hole evaporation, we should assume $\lambda \gg 1$. The joining wormholes are therefore parametrically suppressed, but disks with handles are not since they instead come with factors of $e^{-2gS_0}$ and $S_0$ is not \textit{a priori} a large parameter (in fact, it may have a small or even negative real part).\footnote{In the JT model both geometries are suppressed in the same parameter $e^{-\bfS}$.} This is the analogue of the ``planar" limit in the topological model. There are thus two distinct classes of higher genus geometries: disks with handles and joining wormholes. The higher genus disk geometries can be systematically included in a Schwinger-Dyson equation as in Section \ref{sec:resolvent}, while the joining wormholes can not.

To study the Page curve, we would like to fix the value of $\lambda$ and tune $k$. Since $\lambda \sim e^k$, this involves scaling the prefactor $\fr{e^{2S_0}}{1-e^{-2S_0}}$ with $e^{-k}$. However, this function has a minimum value for real $S_0$, which means the black hole can not evaporate completely. To decrease $\lambda$ beyond the minimum value, we need to go to complex values of $S_0$, namely $e^{2S_0} \in \fr{1}{2} + i \mathbb{R}$. This is a bit strange because it implies a complex action, but is presumably fine because $S_0$ is not a physical parameter (it is not the Bekenstein-Hawking entropy here). This is simply a quirk of the model, and can be attributed as a consequence of having a non-vanishing $S_\partial$.

\section{Discussion}\label{sec:conclusion}

In this paper, we analyzed the behavior of negativity measures in toy models of evaporating black holes in both JT gravity and a topological theory of gravity, with EOW branes.  We found four distinct phases dominated by different saddle-point geometries: the disconnected, cyclically connected, anti-cyclically connected, and pairwise connected.  The last of these geometries are new replica wormholes that break the replica symmetry spontaneously.

We also studied the negativity resolvent using a Schwinger-Dyson equation that resums the contributions of different geometries, and used it to extract the negativity spectrum and negativity measures.  This analysis is valid not only within each of the four phases, but also near phase transitions.  For the topological model or a microcanonical ensemble in JT gravity, we found a cubic equation for the resolvent which can be solved exactly.  For a canonical ensemble in JT gravity, we found a quadratic resolvent equation near the disconnected-pairwise transition, and we solved a more complicated resolvent equation approximately near the cyclic-pairwise transition.  Near this last transition, we found enhanced corrections to various negativity measures: the refined R\'enyi-2 negativity receives an $\cO(1/\sqrt{\b})$ correction, whereas the logarithmic negativity and the partially transposed entropy receive $\cO(1/\b)$ corrections.

These enhanced corrections to negativities are similar to previously found corrections to the von Neumann entropy at the Page transition~\cite{Penington:2019kki,Dong:2020iod,Marolf:2020vsi}.  For the von Neumann entropy, the enhanced corrections can be explained using a diagonal approximation with respect to a basis of fixed-area states~\cite{Dong:2020iod,Marolf:2020vsi,Akers:2018fow}.  We argued that the $\cO(1/\sqrt{\b})$ correction to the refined R\'enyi-2 negativity can be explained in the same way by noting its close connection to von Neumann entropies.  It would be interesting to understand further the $\cO(1/\b)$ corrections to the logarithmic negativity and the partially transposed entropy in a similar way.  Moreover, it would be useful to study the implications of these $\cO(1/\b)$ corrections for the partially transposed entropy more generally: it was conjectured in~\cite{Dong:2021clv} that $S^{T_2}(\r_{R_1R_2})$ is given as a sum of von Neumann entropies $(S_{R_1} +S_{R_2} +S_{R_1 R_2})/2$ in general non-fixed-area states by assuming a diagonal approximation, but this would imply an $\cO(1/\sqrt{\b})$ correction and seems to be in tension with the $\cO(1/\b)$ correction that we find here.

We focused our study on two specific toy models of evaporating black holes, but it would be interesting to generalize our analysis to other models, including the examples studied in~\cite{Almheiri:2019qdq}.

Finally, it would be very interesting to use these results on negativity to diagnose the structure of multipartite entanglement in a realistic evaporating black hole and learn more about its quantum state.  We hope that this will lead to new insights on understanding the interior of black holes and the dynamics of their evaporation.

\section*{Acknowledgments}
We would like to thank Jonah Kudler-Flam, Don Marolf, Henry Maxfield, Geoff Penington, Xiao-Liang Qi, Pratik Rath, Shreya Vardhan, and Michael Walter for useful discussions. This material is based upon work supported by the Air Force Office of Scientific Research under award number FA9550-19-1-0360.

\appendix
\section{Derivation of dominant saddles for negativity}
\label{app:dom}
In this appendix, we derive the set of saddle-point geometries that give dominant contributions to the R\'enyi negativity in various regimes of the parameter space.  This includes each of the distinct phases and near phase transitions.

Our derivation uses facts about geodesics on the permutation group, which we review first. Let $S_n$ be the symmetric group of order $n$, which is the set of permutations on $n$ elements. For any permutation $g \in S_n$, we define $\ell\(g\)$ as the minimum number of swaps from the identity $\mathbbm{1} = (1)(2)\cdots(n)$ to $g$ and $\chi(g)$ as the number of disjoint cycles in $g$, including 1-cycles. These quantities satisfy the relations
\ba
&\ell\(g\) + \chi(g) = n, \label{eq:ellchisum} \\
&\chi(g) = \chi(g^{-1}).
\ea 
As an example, the permutation\footnote{This $g$ is the permutation $12345 \to 21453$ written in cycle notation. Each digit in a given cycle is replaced by the following digit, except for the last digit which is replaced by the first.} $g = (12)(345) \in S_5$ has $\ell\(g\) = 3$ and $\chi(g) = 2$.

We can define the distance between two permutations $g$ and $h$ by
\be\label{eqdist}
d(g,h) \equiv \ell(g^{-1}h)
\ee
which satisfies the usual properties of a distance measure. In particular, given any sequence of permutations $(g_1, \cdots, g_m)$, the distance satisfies the triangle inequality
\be \label{eq:triangle}
d(g_1, g_2) + \cdots + d(g_{m-1},g_m) \geq d(g_1,g_m).
\ee
A sequence of permutations that saturates \eqref{eq:triangle} is said to be on a \textit{geodesic}. We denote a geodesic between two permutations $g$ and $h$ by $G(g, h)$.  We say that a permutation $g'$ is on $G(g, h)$, or equivalently $g'\in G(g,h)$, if the sequence $(g, g', h)$ saturates the triangle inequality \eqref{eq:triangle}.

Our goal is to identify the permutations $g$ that dominate the sum in the R\'enyi negativity~\er{eq:renyi_neg_schem}, in different regimes of the parameter space labeled by $e^{S_0}$, $k_1$, and $k_2$.  We repeat the sum in~\er{eq:renyi_neg_schem} here:
\be\label{eqsum}
\sum_{g \in S_n} \left( e^{S_0}\right)^{\chi(g)} k_1^{\chi(g^{-1} X)} k_2^{\chi(g^{-1} X^{-1})}.
\ee

For reasons that will become clear shortly, it is useful to first identify the permutations on one or more of the following geodesics: $G(\mathbbm{1},X)$, $G(\mathbbm{1},X^{-1})$, and $G(X,X^{-1})$.  Here $X = (12 \cdots n)$ is the cyclic permutation of $n$ elements, and $X^{-1} = (n \cdots 2 1)$ is the anti-cyclic permutation.

For a given permutation $g$, let us use $m$, $p$, $q$ to denote the three exponents in the sum~\er{eqsum}:
\be
m = \chi(g), \qqu
p = \chi(g^{-1} X), \qqu
q = \chi(g^{-1} X^{-1}).
\ee
They satisfy three triangle inequalities, which can be obtained from \eqref{eq:ellchisum}, \er{eqdist}, and~\er{eq:triangle}:
\ba
d(\mathbbm{1},g)+d(g,X) \geq d(\mathbbm{1},X) \qqu &\Rightarrow \qqu m+p \leq n+1,\la{eq1x}\\
d(\mathbbm{1},g)+d(g,X^{-1}) \geq d(\mathbbm{1},X^{-1}) \qqu &\Rightarrow \qqu m+q \leq n+1,\la{eq1xi}\\
d(X,g)+d(g,X^{-1}) \geq d(X,X^{-1}) \qqu &\Rightarrow \qqu p+q \leq n+f(n),\la{eqxxi}
\ea
where $f(n)$ is a useful function defined as
\be
f(n) \eq
\begin{cases}
1, & \text{$n$ odd},\\
2, & \text{$n$ even},
\end{cases}
\ee
and we have used $\chi(X)=\chi(X^{-1})=1$, $\chi(X^2)=f(n)$.

We now identify the permutations $g$ on one or more of the three geodesics.

\paragraph{Permutations on $G(\mathbbm{1},X)$:} These are known to be in one-to-one correspondence with non-crossing partitions, so we say that the corresponding geometries are planar.  We can write such a element as a product of $m$ non-crossing cycles (including 1-cycles):
\be\la{eqgci}
g = \prod_{i=1}^m c_i.
\ee
It is clear that such an element exists for every $m\in [1,n]$.  Since it saturates \er{eq1x}, we immediately find
\be\la{eqpv}
p=n-m+1.
\ee
Moreover, it is straightforward to derive
\be\la{eqqv}
q = \sum_{i=1}^m f(|c_i|) - m +1,
\ee
where $|c_i|$ is the length of the $i$-th cycle $c_i$.

\paragraph{Permutations on $G(\mathbbm{1},X^{-1})$:} These can be obtained by simply taking the inverse of the permutations on $G(\mathbbm{1},X)$, sending $c_i$ in \er{eqgci} to $c_i^{-1}$.  We say that these correspond to ``anti-planar'' geometries.

\paragraph{Permutations on $G(\mathbbm{1},X)$ and $G(\mathbbm{1},X^{-1})$:} Their cycles $c_i$ must be their own inverses, so the length of each cycle is at most 2.  Therefore, these permutations are precisely those non-crossing partitions that consist of only 1-cycles and 2-cycles.  Such permutations exist for every $m \geq \lceil \frac n2 \rceil$, with the lower bound saturated by non-crossing pairings consisting of $\lceil \frac n2 \rceil$ pairs and at most one 1-cycle.

\paragraph{Permutations on $G(\mathbbm{1},X)$ and $G(X,X^{-1})$:} As they saturate \er{eqxxi}, it is straightforward to use \er{eqpv} and \er{eqqv} to show that these permutations are precisely those non-crossing partitions with at most one odd cycle.  Here we define an odd cycle as one of odd length and an even cycle as one of even length.  For even $n$, these permutations consist of only even cycles, whereas for odd $n$, they have exactly one odd cycle.  Such permutations exist for every $m \leq \lceil \frac n2 \rceil$, with the upper bound saturated by non-crossing pairings.

\paragraph{Permutations on $G(\mathbbm{1},X^{-1})$ and $G(X,X^{-1})$:} These are obtained by taking the inverse of the permutations on $G(\mathbbm{1},X)$ and $G(X,X^{-1})$.

\paragraph{Permutations on $G(\mathbbm{1},X)$, $G(\mathbbm{1},X^{-1})$, and $G(X,X^{-1})$:} It is clear by combining the previous cases that these permutations are those non-crossing partitions that consist of only 2-cycles and at most one 1-cycle.  Therefore, they are in one-to-one correspondence with non-crossing pairings~\cite{Dong:2021clv}.  These all have $m = \lceil \frac n2 \rceil$ and $p = q = \lfloor \frac n2 \rfloor+1$.  We denote these non-crossing pairings by $\t$, and say that they correspond to pairwise geometries.  A simple example is $\tau = (12)(34)\cdots(n-1,n)$ for even $n$ and $\tau = (12)(34)\cdots(n-2,n-1)(n)$ for odd $n$.\\

\begin{figure}
    \centering
    \includegraphics[width=.5\textwidth]{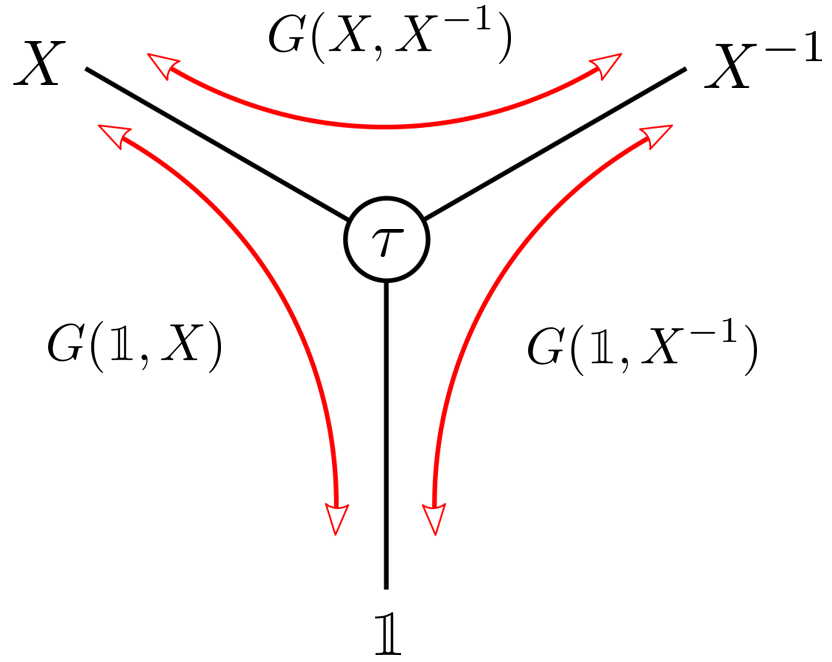}
    \caption{A schematic Cayley graph for the relevant permutations and the geodesics connecting them.}
    \label{fig:cayley}
\end{figure}

These results are illustrated schematically in Figure \ref{fig:cayley}. We now state and prove the main points of this appendix.

\begin{lemma}\la{lmplanar}
In the regime $k_1/k_2 \gg e^{-S_0}$, planar geometries dominate the sum~\er{eqsum}.  In other words, for any $g \notin G(\mathbbm{1},X)$, there exists $g' \in G(\mathbbm{1},X)$ such that $g'$ dominates over $g$.
\end{lemma}
\begin{proof}
As $g \notin G(\mathbbm{1},X)$, the difference
\be
d(\mathbbm{1},g)+d(g,X) - d(\mathbbm{1},X)
\ee
is positive.  However, this difference must be even, regardless of whether $g$ is an even or odd permutation.  Therefore, we denote this difference by $2r$ with a positive integer $r$, and use it to rewrite the triangle inequality~\er{eq1x} as
\be\la{eqmp}
m+p = n+1-2r.
\ee
Our goal is to choose a more dominant $g'$.  For $g'$, we define $m'$, $p'$, $q'$ similarly as
\be
m' = \chi(g'), \qqu
p' = \chi(g'^{-1} X), \qqu
q' = \chi(g'^{-1} X^{-1}).
\ee

Let us discuss $m+r \geq \lceil \frac n2 \rceil$ and $m+r < \lceil \frac n2 \rceil$ separately.  For $m+r \geq \lceil \frac n2 \rceil$, we choose $g'$ to be on $G(\mathbbm{1},X)$ and $G(\mathbbm{1},X^{-1})$, with $m'=m+r$.  As we discussed earlier, such a permutation exists --- in particular, \er{eqmp} guarantees $m + r < n$.  As $g'$ saturates \er{eq1x} and \er{eq1xi} (after primes are added), we find
\be
p'=n-m-r+1 = p+r,\qqu
q'=n-m-r+1 \geq q-r,
\ee
where the second equality for $p'$ comes from \er{eqmp} and the inequality comes from \er{eq1xi} for $g$.  We thus find that $g'$ gives a more dominant contribution to the sum~\er{eqsum} than $g$, as
\be\la{eqratio}
\fr{e^{m' S_0} k_1^{p'} k_2^{q'}}{e^{m S_0} k_1^{p} k_2^{q}} \geq \( e^{S_0}\fr{k_1}{k_2} \)^r \gg 1.
\ee

In the other case with $m+r < \lceil \frac n2 \rceil$, we choose $g'$ to be on $G(\mathbbm{1},X)$ and $G(X,X^{-1})$, with $m'=m+r$.  Again, such a permutation exists.  As $g'$ saturates \er{eq1x} and \er{eqxxi} (after primes are added), we find
\be
p'=n-m-r+1 = p+r,\qqu
q'=n+f(n)-p-r \geq q-r,
\ee
where the inequality comes from \er{eqxxi} for $g$.  We again find that \er{eqratio} holds and therefore $g'$ dominates over $g$.
\end{proof}

It is clear that Lemma \re{lmplanar} is tight in the sense that every planar geometry could give a dominant contribution to the sum~\er{eqsum} at some point in the regime $k_1/k_2 \gg e^{-S_0}$.  In particular, at the point where $k_1 = e^{S_0}$ and $k_2=1$, they give an equal contribution $e^{(n+1)S_0}$.

From Lemma \re{lmplanar}, we immediately obtain the following corollary by taking the inverse of all permutations and switching $k_1 \lra k_2$.

\begin{corollary}\la{lmanti}
In the regime $k_2/k_1 \gg e^{-S_0}$, anti-planar geometries dominate the sum~\er{eqsum}.
\end{corollary}

Combining Lemma \re{lmplanar} and Corollary \re{lmanti}, and recalling that the permutations on $G(\mathbbm{1},X)$ and $G(\mathbbm{1},X^{-1})$ are precisely those that consist of only 1-cycles and 2-cycles, we immediately obtain the following corollary (which is useful for studying the disconnected-pairwise transition in Section~\re{sec:dispair}).
 
\begin{corollary}\la{lmdispair}
In the regime $e^{-S_0} \ll k_1/k_2 \ll e^{S_0}$, the permutations consisting of only 1-cycles and 2-cycles dominate the sum~\er{eqsum}.
\end{corollary}

It is again clear that Corollary \re{lmdispair} is tight in the sense that every permutation on $G(\mathbbm{1},X)$ and $G(\mathbbm{1},X^{-1})$ could give a dominant contribution to the sum~\er{eqsum} at some point in the regime $e^{-S_0} \ll k_1/k_2 \ll e^{S_0}$.  In particular, at the point where $k_1=k_2=e^{S_0/2}$, they all give an equal contribution $e^{(n+1)S_0}$.

\begin{lemma}
In the regime $k_1 k_2 \ll e^{S_0}$, the disconnected geometry dominates the sum~\er{eqsum}.
\end{lemma}
\begin{proof}
The disconnected geometry is represented by the identity $\mathbbm{1}$ and contributes $e^{nS_0}k_1k_2$.  For any other permutation $g$, we have $m \leq n-1$.  From \er{eq1x} and \er{eq1xi}, we obtain $p, q \leq n-m+1$.  We thus find that $\mathbbm{1}$ gives a more dominant contribution to the sum~\er{eqsum} than $g$, as
\be
\fr{e^{n S_0} k_1 k_2}{e^{m S_0} k_1^{p} k_2^{q}} \geq \(\fr{e^{S_0}}{k_1 k_2}\)^{n-m} \gg 1.
\ee
\end{proof}

\begin{lemma}\la{lmcyclic}
In the regime $k_1/k_2 \gg e^{S_0}$, the cyclic geometry dominates the sum~\er{eqsum}.
\end{lemma}
\begin{proof}
The cyclic geometry is represented by $X$ and contributes $e^{S_0} k_1^n k_2^{f(n)}$.  For any other permutation $g$, we have $p \leq n-1$.  From \er{eq1x} and \er{eqxxi}, we obtain $m \leq n-p+1$ and $q \leq n+f(n)-p$.  We thus find that $X$ gives a more dominant contribution to the sum~\er{eqsum} than $g$, as
\be
\fr{e^{S_0} k_1^n k_2^{f(n)}}{e^{m S_0} k_1^{p} k_2^{q}} \geq \(\fr{k_1}{k_2} e^{-S_0}\)^{n-p} \gg 1.
\ee
\end{proof}

From Lemma \re{lmcyclic}, we immediately obtain the following corollary by taking the inverse of all permutations and switching $k_1 \lra k_2$.

\begin{corollary}
In the regime $k_2/k_1 \gg e^{S_0}$, the anti-cyclic geometry dominates the sum~\er{eqsum}.
\end{corollary}

We now show that pairwise geometries dominate a fourth phase.  We first derive the following lemma as a useful intermediate step.

\begin{lemma}\la{lmpost}
In the regime $k_1 k_2 \gg e^{S_0}$, the permutations on $G(X,X^{-1})$ dominate the sum~\er{eqsum}.  In other words, for any $g \notin G(X,X^{-1})$, there exists $g' \in G(X,X^{-1})$ such that $g'$ dominates over $g$.
\end{lemma}
\begin{proof}
As $g \notin G(X,X^{-1})$, the triangle inequality~\er{eqxxi} must fail to saturate by a positive but even integer, which is at least 2:
\be\la{eqpq}
p+q \leq n+f(n)-2 = 2 \lt\lfloor \frac n2 \rt\rfloor.
\ee
Therefore, one of $p$, $q$ must be no greater than $\lfloor \frac n2 \rfloor$.  Without loss of generality, we consider the case of $p \leq \lfloor \frac n2 \rfloor$.  We then choose $g'$ to be on $G(\mathbbm{1},X)$ and $G(X,X^{-1})$, with $p'=p+1$.  As $g'$ saturates \er{eq1x} and \er{eqxxi} (after primes are added), we find
\be
m'=n-p \geq m-1,\qqu
q'=n+f(n)-p-1 \geq q+1,
\ee
where the two inequalities comes from \er{eq1x} and \er{eqpq}, respectively.  As we discussed earlier, such a permutation $g'$ exists, as $m' = n-p \geq \lceil \frac n2 \rceil$.  From this, we find that $g'$ gives a more dominant contribution to the sum~\er{eqsum} than $g$, as
\be
\fr{e^{m' S_0} k_1^{p'} k_2^{q'}}{e^{m S_0} k_1^{p} k_2^{q}} \geq \fr{k_1 k_2}{e^{S_0}} \gg 1.
\ee
\end{proof}

Combining Lemmas~\re{lmplanar}, \re{lmpost} and Corollary~\re{lmanti}, and recalling that the permutations on all three geodesics $G(\mathbbm{1},X)$, $G(\mathbbm{1},X^{-1})$, and $G(X,X^{-1})$ are precisely non-crossing pairings that lead to pairwise geometries, we immediately obtain the following corollary.

\begin{corollary}\la{lmpair}
In the regime satisfying both $k_1 k_2 \gg e^{S_0}$ and $e^{-S_0} \ll k_1/k_2 \ll e^{S_0}$, the pairwise geometries dominate the sum~\er{eqsum}.
\end{corollary}

It is clear that Corollary~\re{lmpair} is tight in the sense that all pairwise geometries give an equal, dominant contribution $( e^{S_0} )^{\lceil \frac n2 \rceil} (k_1 k_2)^{\lfloor \frac n2 \rfloor + 1}$ to the sum~\er{eqsum}, as they all have the same $m = \lceil \frac n2 \rceil$ and $p = q = \lfloor \frac n2 \rfloor+1$.

\section{Details of the cyclic-pairwise transition in the canonical ensemble}
\label{app:trans}

In Section \ref{sec:trans}, we used some of the techniques developed in \cite{Penington:2019kki} to derive an approximation for the eigenvalue spectrum of the partially transposed density matrix near transition:
\be
D(\lambda) = e^{S_0} \int_0^{s_k} ds \rho(s) \left[ \frac{k_2(k_2+1)}{2} \delta\left(\lambda - \lambda_0 - \frac{w(s)}{k_2}\right) + \frac{k_2(k_2-1)}{2} \delta \left(\lambda - \lambda_0 + \frac{w(s)}{k_2} \right) \right].
\ee
This approximation was derived under a set of assumptions which we repeat here:
\begin{enumerate}
    \item $w(s_t)R \ll k k_2$
    \item $k_2^2 e^{S_0} \int_0^{s_t} ds \rho(s) \frac{w(s)R(k + w(s)R)}{k^2k_2^2 - w(s)^2R^2} \ll k$
    \item $s_t = s_k - \kappa$, where $\kappa$ is $\mathcal{O}(1)$ but large
\end{enumerate}
Our analysis in this section will be based on checking the consistency of the iterative procedure we applied to the resolvent equation \eqref{eq:mideq}, namely the zeroth order approximation
\be
R_0 = \frac{k}{\lambda - \lambda_0}.
\label{eq:R0}
\ee
We start with Assumption (2). We want to rigorously show the following inequality on the second term in \eqref{eq:mideq}:
\be
\abs{ k_2^2 e^{S_0} \int_0^{s_t} \rho(s) \frac{w(s)R(k + w(s)R)}{k^2 k_2^2 - w(s)^2 R^2}} \equiv  \abs{ e^{S_0} \int_0^{s_t}  \rho(s) f(s)} \ll k.
\ee
This function has a pole located at $s = s_*$, which is captured by the integral under the assumption $\abs{w(s_t)R} \ll kk_2$, as $w(s)$ is a monotonically decreasing function of $s$. We can therefore rewrite the integral with an $i \epsilon$ prescription as
\be
e^{S_0}\int ds \rho(s)f(s) = PV \left( e^{S_0}\int ds \rho(s)f(s) \right) \pm i \pi k_2^2 e^{S_0} \int ds \rho(s) \frac{w(s)R(k+w(s)R)}{\partial_s\left( k^2 k_2^2 - w(s)^2 R^2 \right)} \delta( s - s_*),
\label{eq:llk}
\ee
where PV denotes the Cauchy principal value. We choose the sign of $i \epsilon$  arbitrarily, as we are only looking to bound the absolute value of this integral. 

Let us treat the first term. The principle value is dominated by the $(s-s_*)^0$ term in the Laurent expansion of $\rho(s)f(s)$. We can perform a Laurent expansion around $s = s_*$ using the semiclassical approximation
\be
w(s) \approx e^{-\beta s^2/2 - S_0 - 2\pi^2/\beta}.
\ee
We find
\be
\rho(s)f(s) = \rho(s_*)\left(\frac{k_2(k_2+1)}{2(\beta s_*)(s-s_*)} - \frac{k_2(1+k_2(1+2\beta s_*^2))}{4 \beta s_*^2} + \mathcal{O}(s-s_*)\right).
\ee
Ignoring scaling in $\beta$, the $\mathcal{O}(s-s_*)^0$ term goes like $k_2^2 \sim k/e^{S_0} \ll k$, so we can safely ignore this term.

Let us now look at the second term in \eqref{eq:llk}. We have
\ba
\left. k_2^2 e^{S_0} \rho(s) \frac{w(s)R(k+w(s)R)}{\partial_s\left( k^2 k_2^2 - w(s)^2 R^2 \right)} \right\vert_{s = s_*} &\approx \left. k_2^2 e^{S_0} \rho(s) \frac{(w(s)R)^2}{(2\beta s)w(s)^2R^2} \right\vert_{s = s_*} \nonumber \\
&\approx   \frac{k_2^2 e^{S_0} \rho(s_*)}{2 \beta s_*},
\ea
where we have used $\abs{w(s_*)R} = kk_2 \gg k$ to simplify the numerator. We can rewrite this in terms of $s_k$ using \eqref{eq:skdef} and the asymptotic form of $\rho(s)$:
\ba
\frac{k_2^2 e^{S_0} \rho(s_*)}{2 \beta s_*} &\approx \frac{k_2^2}{2 \beta}e^{2 \pi s_*} \nonumber \\
&\approx \frac{k_2^2}{2 \beta s_k}e^{S_0} \rho(s_k)e^{2 \pi (s_*-s_k)} \nonumber \\
&= \frac{k}{2 \beta s_k}e^{2 \pi (s_*-s_k)} \ll k.
\ea
As stated previously, $s_k \sim \mathcal{O}(1/\beta)$, so for this to be much smaller than $k$ we require $s_k - s_*$ being at least $\mathcal{O}(1)$ but large, and by proxy $\kappa \equiv s_k - s_t$ being at least $\mathcal{O}(1)$ but large, as stated in assumption (3).

What is stopping $\kappa$ from being much larger, say $\mathcal{O}(1/\beta)$? Now we check the validity of assumption (1), that is $\abs{w(s_t)R} \ll kk_2 $. Under our approximation \eqref{eq:R0}, this assumption translates into the condition
\be
\abs{\lambda - \lambda_0} \gg \frac{w(s_t)}{k_2}.
\ee
However, at the boundary of our spectrum located at $s = s_t$, we should also have $\abs{\lambda-\lambda_0} = w(s_t)/k_2$. The way to reconcile these two assumptions is to state that our approximation only holds for some range of $s$ between $0$ and some control parameter $s_c$ where 
\be
\abs{\lambda - \lambda_0} = \frac{w(s_c)}{k_2}.
\ee
What is the value of $s_c$? We have
\be
w(s_c) \gg w(s_t) \Rightarrow e^{-\beta(s_t^2-s_c^2)/2} \ll 1 
\label{eq:sc1}.
\ee
We define $\kappa' \equiv s_t - s_c$. Now \eqref{eq:sc1} becomes
\be
2 s_t \kappa' + \kappa'^2 \gg 1/\beta.
\ee
Under our previous assumption that $\kappa$ is $\mathcal{O}(1)$ in $\beta$ such that $s_t \sim 1/\beta$, $\kappa'$ too must be $\mathcal{O}(1)$ but large. This answers our previous question about the size of $\kappa$. There is something of an inverse relationship between $s_t$ and $\kappa'$: our goal is to design an approximation in which $s_c$ is as close as possible to $s_t$, as well as one in which $s_t$ is as close as possible to $s_k$. Under the assumption that $s_t$ is as close as possible to $s_k$, that is $\kappa$ is $\mathcal{O}(1)$ but large, we also have $\kappa'$ $\mathcal{O}(1)$ but large. If $s_t$ was far from $s_k$ so that, say, $s_t \sim 1/\sqrt{\beta}$, we would also require $\kappa' \sim \mathcal{O}(1)/\sqrt{\beta}$ with a large $\mathcal{O}(1)$ constant, thus missing a large part of the spectrum in our approximation.

Our conclusion is that there exists a region of size $\mathcal{O}(1)$ in $s$ where assumption (1) does not hold. As the spectrum is over a region of size $\frac{w(0) - w(s_k)}{k_2}$, the size of a region $\mathcal{O}(1)$ in $s$ is exponentially suppressed in $1/\beta$, and we conclude that very few eigenvalues are in the uncontrolled region.

\section{R\'enyi entropies near the Page transition}\label{app:renyi}

We draw an analogy between the $\mathcal{O}(1/\beta)$ corrections to the logarithmic negativity and the partially transposed entropy and the $\mathcal{O}(1/\beta)$ corrections to the R\'enyi entropy $S_n$ with $n < 1$. Here we show this result explicitly in the model of \cite{Penington:2019kki}. We recall many of their results, which can equivalently be obtained from ours by sending $k_1$ to $k$ and $k_2$ to $1$.

We consider the model of Section \ref{sec:jt}, but without partitioning the radiation system. The approximation for the density of states \eqref{eq:NegDApprox} is now
\be
D(\lambda) = e^{S_0} \int_0^{s_k} ds \rho(s) \delta( \lambda - \lambda_0 - w(s)),
\ee
where $w(s)$ and $\lambda_0$ are defined as in the main text. Here $s_k$ is defined as
\be
k = e^{S_0}\int_0^{s_k} ds \rho(s) \Rightarrow s_k \approx \frac{1}{2 \pi} \left( \log k - S_0 \right).
\ee
The transition at the Page time can be thought of as the transition from the fully disconnected phase to the cyclic phase along the $x$ axis of our phase diagram (Figure \ref{fig:phase1}). Using our results from Table \ref{table} and the semiclassical approximation $\mu \gg \frac{1}{\beta} \gg 1$, we have
\ba
\Tr \rho_R^n &= \left(\frac{Z_1}{k}\right)^{n-1} = \frac{Z_n}{Z_1^n} \nonumber \\
\Rightarrow S_n &= S_0 + \left(1 + \frac{1}{n}\right) \frac{2 \pi^2}{\beta}
\label{eq:Snappendix}.
\ea
Again we can solve for $\log k$ at transition in the semiclassical regime to obtain
\be
\log k =  \log \left( \frac{Z_1^n}{Z_n} \right)^{\frac{1}{n-1}} = S_0 + \left( 1 + \frac{1}{n} \right) \frac{2 \pi^2}{\beta}.
\label{eq:logkRenyi}
\ee
From this we find $s_k^{(n)}$ at transition to be
\be
s_k^{(n)} = \frac{\pi}{\beta} \left( 1 + \frac{1}{n} \right).
\ee
As $s^{(n)} = \frac{2\pi}{n \beta}$, for $n < 1$ we have $s_k^{(n)} < s^{(n)}$. Part of our derivation relied on $s_k$ scaling like $1/\beta$, which remains true for $n$ of $\mathcal{O}(1)$.

The R\'enyi entropy is given by
\ba
S_n &= \frac{1}{1-n} \log \int_{-\infty}^{\infty} d\lambda D(\lambda) \lambda^n \nonumber \\
&= \frac{1}{1-n} \log \left( e^{S_0} \int_0^{s_k} ds \rho(s) \left( \lambda_0 + w(s) \right)^n \right).
\ea
As $s_k > s^{(1)}$, $\lambda_0$ will be exponentially suppressed in $1/\beta$, so this integral is dominated by the $w(s)$ term, and as $s_k^{(n)} < s^{(n)}$ we approximate
\ba
S_n &\approx \frac{1}{1-n} \log \left( e^{S_0} \int_0^{s_k} ds \rho(s) w(s)^n \right)
\nonumber \\
&\approx \frac{1}{1-n} \log \left( e^{S_0} \rho(s_k) w(s_k)^n \right) \nonumber \\
&\approx \fr{1}{1-n}\( \log k - n S_0 - \fr{n\beta s_k^2}{2}-\fr{2\pi^2 n}{\beta}\).
\ea
Using our expressions \eqref{eq:logkRenyi} and $s_k^{(n)}$, we find
\be
S_n = S_0 + \left( \frac{3 + 5n}{2n}\right) \frac{\pi^2}{\beta}.
\ee
Comparing this to our previous answer \eqref{eq:Snappendix}, we find a correction $\Delta S_n$ at transition of the form
\be
\Delta S_n = \frac{\pi^2}{2\beta} \left( 1 - \frac{1}{n}\right).
\ee
We conclude that there are enhanced corrections of the form $\mathcal{O}(1/\beta)$ to the R\'enyi entropy for $n<1$. 

\bibliographystyle{JHEP}
\bibliography{bibliography}

\end{document}